\documentclass[journal]{IEEEtran}
\usepackage[pdftex]{graphicx}
\graphicspath{{../pdf/}{../jpeg/}}
\DeclareGraphicsExtensions{.pdf,.jpeg,.png}
\usepackage[cmex10]{amsmath}
\usepackage{graphicx,color}
\usepackage{graphicx}
\usepackage{amssymb,color}
\allowdisplaybreaks
\usepackage{subcaption}
\usepackage{cite}
\usepackage{blindtext}
\usepackage{hyperref}
\usepackage{tcolorbox}
\usepackage{xcolor}
\usepackage{lipsum}
\usepackage[font={small}]{caption}
\usepackage{algorithm}
\usepackage{algorithmic}
\usepackage{array} 
\usepackage{makecell}
\usepackage{tabularx}
\newtheorem{theorem}{Theorem}

\newtheorem{proof}{Proof}
\newtheorem{remark}{Remark}
\newtheorem{proposition}{Proposition}

\begin{document}

\title{User Scheduling and Trajectory Optimization for Energy-Efficient IRS-UAV Networks with  SWIPT}
\author{S. Zargari, A. Hakimi,  C. Tellambura, \textit{Fellow, IEEE}, and S. Herath,~\IEEEmembership{Member,~IEEE}
\thanks{This work was supported in part by Huawei Technologies Canada Company, Ltd.}
\thanks{Shayan Zargari, Azar Hakimi, and Chintha Tellambura are with the Department of Electrical and Computer Engineering, University of Alberta, Edmonton, AB T6G 1H9, Canada (e-mail: zargari@ualberta.ca; hakimina@ualberta.ca; ct4@ualberta.ca).}
\thanks{ Sanjeewa Herath is with the Ottawa Wireless Advanced System Competency Centre, Huawei Canada, Ottawa, ON K2K 3J1, Canada (e-mail: sanjeewa.herath@huawei.com).}\vspace{-3mm}}
\maketitle

\begin{abstract}
This paper investigates user scheduling and trajectory optimization for a network supported by an intelligent reflecting surface (IRS) mounted on an unmanned aerial vehicle (UAV). The IRS is powered via the simultaneous wireless information and power transfer (SWIPT)  technique. The IRS boosts users' uplink signals to improve the network's longevity and energy efficiency. It simultaneously harvests energy with a non-linear energy harvesting circuit and reflects the incident signals by controlling its reflection coefficients and phase shifts. The trajectory of the UAV impacts the efficiency of these operations. We minimize the maximum energy consumption of all users by joint optimization of user scheduling, UAV trajectory/velocity, and IRS phase shifts/reflection coefficients while guaranteeing each user's minimum required data rate and harvested energy of the IRS. We first derive a closed-form solution for the IRS phase shifts and then address the non-convexity of the critical problem. Finally, we propose an alternating optimization (AO) algorithm to optimize the remaining variables iteratively. We demonstrate the gains over several benchmarks. For instance, with a  50-element IRS, min-max energy consumption can be as low as 0.0404 (Joule),  a 7.13\%  improvement over the No IRS case  (achieving 0.0435 (Joule)). We also show that IRS-UAV without EH performs best at the cost of circuit power consumption of the IRS (a 20\% improvement over the No IRS case).
\end{abstract}
\section{Introduction}

\IEEEPARstart{T}{he} sixth-generation (6G) wireless will appear around 2030\cite{Yang6G}, and intelligent reflecting surface (IRS), also called reconfigurable intelligent surface (RIS), is one of the critical enabling technologies  \cite{Gong1}. IRS comprises many low-cost reflecting elements. They can intelligently adjust the incident signals' phase and amplitude to enhance the received signal strength at the receiver \cite{Zhang}. IRSs can be mounted on walls and similar places,  increasing the spatial degrees of freedom. Thus, the wireless channel itself becomes an optimization variable. IRSs are nearly passive devices,  requiring a small amount of energy, which opens the possibility of energy harvesting (EH) to power them.  

Likewise, unmanned aerial vehicles  (UAVs),  also known as aerial base stations (BS), offer another 6G wireless option. UAVs augment the coverage area and provide reliable communication along with fast and on-demand deployment, which is suitable for relaying, information gathering, and data distribution \cite{ Bariah}. However, the drawbacks of these UAV schemes are a few:  

\begin{enumerate}
\item  The limited onboard power of UAVs restricts their service time  \cite{Yang6G,Khalili2021Hetnets}. In particular, EH techniques such as simultaneous wireless information and power transfer (SWIPT) may help  \cite{yu2021,zargari2,zargari3,zargari4,zargari2021robust}. However,  harvesting sufficient energy may take a long time \cite{Ramezani}.

\item The power consumption of the UAV is for the communication tasks and propulsion energy for hovering and supporting high mobility over the air. The latter is usually several orders of magnitude higher than the former. This issue imposes critical limits on communication performance \cite{you2021enabling}.

\item  Relays operate in either decode-and-forward (DF) or amplify-and-forward (AF) modes, reducing the overall spectral efficiency (SE). This deficiency may be addressed using full-duplex (FD) relays \cite{Mohammadi2015}. However,  complex self-interference cancellation (SIC)  techniques are needed to mitigate the SI \cite{Zhang:2015}.
\end{enumerate}

The challenges mentioned above strongly motivate the integration of UAVs and IRS. Compared to conventional relays, IRS is more energy-efficient \cite{Gong1}. Furthermore, IRS is inherently FD, does not need SIC, and does not add noise. All these reasons suggest the benefits of  IRS/UAV networks. However, there are two ways of doing that. 
\begin{enumerate}
\item[i)] Terrestrial IRS UAV networks:  Many existing works (see  \cite{Pan2, Pan1,Kwan} and references therein)  consider a terrestrial IRS at a fixed location and the UAV as an aerial BS. However, the UAV as an aerial BS has stringent size, weight, and power constraints, which impose critical limits on its flight time or endurance.   

\item[ii)] IRS-UAV networks: These deploy an IRS on the  UAV \cite{JiaoFZZ20,Zina_Mohamed,Xiaowei_Pang}, which establishes line-of-sight (LoS) links with the ground-level users. Hereafter, the trajectory of the IRS-UAV can be more flexibly optimized. Also, the IRS-UAV  can achieve $360^o$ panoramic full-angle reflections, i.e., one IRS-UAV can manipulate signals between any pair of ground nodes.
\end{enumerate}

{We now further elaborate on the benefits of the IRS-UAV concept.  First, the onboard power of the UAV limits its operating time as an aerial BS.  Instead, with an IRS,  the UAV can serve users by reflection signals.  This may save the UAV  hovering and flight energy consumption \cite{Khalili2021Hetnets}. Second, IRS-UAV operates in the  FD mode, which further extends the effective SE \cite{Gong1}. Third, deploying multiple antennas at the UAV may increase the cost \cite{{Ramezani}}. However, an IRS comprises reflecting elements that provide spatial gains. The advantages of the IRS, such as being lightweight, small size, low energy consumption, and cohesive geometry will enhance the performance of UAV-only networks.} 

On the other hand, the energy consumption of the IRS  is often assumed negligible due to its passive reflection mode. However, this assumption is not entirely correct as this energy consumption increases with the number of reflecting elements,   a critical issue \cite{Gong1}. For example, for $3$ and $5$-bit resolution phase shifting, each elements consumes about  $1.5$ to $6$ mW, respectively~\cite{Hu2021}. Furthermore,  active IRS has emerged recently, which can amplify the incident signals \cite{chen2022active}, albeit at the cost of increased power consumption.  Because energy efficiency is a key performance indicator in 6G wireless, we investigate potential EH benefits using  SWIPT. 

\subsection{Contributions}

To overcome the limitations above,   this work investigates the IRS-UAV  network (Fig. 1).   In particular, the IRS is powered by means of EH  \cite{Hu2021, Ramezani}. This improves overall energy efficiency (EE) yet does not cost extra power signals or resources because EH is based on the SWIPT technique.   For a  more realistic model, we assume that the EH process follows a non-linear model, unlike many other published works. Note that in our system setup (Fig. 1),  the function of the UAV is to move the IRS   close to the users (via three-dimensional (3D) trajectory optimization). We investigate the uplink transmission because that allows the power minimization of mobile users, which is critical to extending their battery life \cite{Jae_Cheol}. Indeed, the IRS-UAV establishes a reflection link that increases the throughput and reduces the transmit power of the users \cite{zargari2,Zheng2,zargari,Wei,Kwan}. Therefore, performing EH and reflecting signals simultaneously at the IRS are conflicting objectives affecting users' transmission power.  When the IRS uses power splitting (PS)-based SWIPT,   a portion of harvested power is used for ID and the remainder for the reflection. However,  since users are transmitting in the uplink, their transmit power affects the EH/reflection performance of the IRS. \textbf{Overall, this paper studies a new approach to realizing energy-efficient IRS-UAV networks.}

The main contributions of this work are stated as follows.

\begin{itemize}
\item Unlike  \cite{Pan2, Pan1,Kwan}, which study separate IRS and UAV networks, we investigate IRS-UAV networks capable of EH using the PS architecture. The PS circuit diverts a  fraction of incident RF power at the IRS  to the EH circuit. We also improve the EH process modeling because previous studies  {\cite{JiaoFZZ20,Zina_Mohamed,Xiaowei_Pang,zargari,Zheng1,Zheng2,zargari2,zargari4,Gong,Ramezani,yu2021}}  adopt a  simplified linear EH model. Linear models do not capture the two most significant nonlinear effects of actual EH circuits. First, the harvested power saturates as the input power grows large. Second, harvested power drops to zero as the input power falls below the sensitivity level of the EH circuit. To ease these issues,  we adopt a non-linear EH model at the IRS, which results in a fundamentally different optimization problem \cite{zargari3,Nonlinear}.

\item While \cite{Pan2,Pan1,mahmoud2021,Nnamani,Zhang3} optimize the UAV trajectory in a two-dimensional (2D) plane,  we consider the altitude of the UAV in  3D space. On the other hand, in practice, the UAV velocity is variable  \cite{Pan2,Pan1,Kwan,mahmoud2021,Taniya2021,Nnamani,Zhang3}. Accordingly, we also study the behavior of the UAV velocity under different trajectories.

\item We optimize user scheduling and trajectory optimization to ensure that the energy consumption of all users is low. To this end, we minimize the maximum energy consumption of users by jointly optimizing the user scheduling, UAV trajectory/velocity, and IRS phase shifts/reflection coefficient while guaranteeing the minimum data rate and energy requirement of each user and the IRS, respectively.

\item To attack this highly coupled optimization problem, we propose an alternating optimization (AO) algorithm to decouple optimization variables. The basic idea of AO is to break down the overall problem into simpler subproblems and solve the sequentially until convergence. By following this approach, we first derive a closed-form solution for the IRS phase shifts. Second, we formulate three sub-optimal problems and obtain sub-optimal solutions for the user scheduling, UAV trajectory/velocity, and IRS reflection coefficients, respectively. Finally, we leverage the difference of convex (DC) programming and successive convex approximation (SCA) to tackle the last two sub-optimal problems. While for the user scheduling, we perform a continuous relaxation of the binary constraint into a continuous interval.

\item We present numerical results to evaluate and assess the performance of the proposed network. We find that it outperforms baseline schemes in terms of the energy consumption of users except for the case of No EH at the IRS.

\end{itemize}

\subsection{Related Works}

The related works fall into three groups. We next discuss each group separately; namely,  i)  Terrestrial IRS UAV networks; ii) IRS-UAV networks;  iii) EH technologies for  IRS networks.

\textit{i) \textbf{Terrestrial IRS UAV networks:}} In these,  the IRS is placed anywhere on the ground  (terrestrial IRS) and assists UAV communication \cite{Pan2, Pan1,Kwan}. For instance,  \cite{Pan2}  maximizes the minimum average data rate of the IRS by joint design of the UAV trajectory and IRS phase shifts/scheduling. The authors in \cite{Pan1} consider an UAV-aided IRS network supporting terahertz (THz) communications where the minimum average achievable data rate of all users is maximized by jointly optimizing the IRS phase shifts, UAV trajectory, and power as well as THz sub-bands allocation. Reference  \cite{Kwan} considers an IRS-UAV-based orthogonal frequency division multiple (OFDM) access network and maximizes the network sum rate by jointly optimizing the UAV trajectory, IRS phase shifts, and resource allocations.

{\textit{ii) \textbf{IRS-UAV networks:}}
In these, the IRS is mounted on the UAV \cite{Taniya2021,yu2021, Nnamani, Zhang3, Khalili2021Hetnets, JiaoFZZ20, Zina_Mohamed, Xiaowei_Pang}. In  \cite{Taniya2021},  an IRS-UAV network is studied along with only-IRS and only-UAV modes. The network outage probability, ergodic capacity, EE, and optimization over some critical network parameters are derived. Also, the symbol error rate and outage probability are analyzed for an IRS-UAV assisted network in \cite{mahmoud2021}. In \cite{yu2021}, an IRS-UAV downlink network is considered where the IRS-UAV supports the energy and information transmission to a single  Internet of Things (IoT) device. This work maximizes the average achievable data rate by optimizing the transmit beamforming, PS ratio, IRS phase shifts, and UAV trajectory.  In \cite{Nnamani}, the achievable secrecy rate of an IRS-UAV network is maximized by jointly optimizing the UAV trajectory, IRS phase shifts, and collaborative beamformers at the sensor nodes. In \cite{Zhang3}, a new 3D wireless relaying network aided by an aerial IRS is proposed where the worst-case signal-to-noise ratio (SNR) is maximized by joint design of the location as well as 3D passive beamforming at the aerial IRS and active beamformers at the source node. Furthermore, to have an energy-efficient network,  multiple UAV-IRS objects are introduced in a macro-cell power domain non-orthogonal multiple access (NOMA) Heterogeneous network~\cite{Khalili2021Hetnets}.  Finally, the resource allocation for transmitting power minimization problem is considered where dueling deep Q-Network learning and semi-definite relaxation techniques (SDR) are employed for this purpose. Besides,  \cite{JiaoFZZ20} considers an IRS-UAV network that operates based on the NOMA. The rate of the strong user is maximized by joint design of the UAV horizontal location, beamforming vectors, and IRS phase shift matrix. To enable energy-efficient communications with cell-edge users, \cite{Zina_Mohamed} optimizes an IRS-UAV network. In addition, a comparison between the achievable rate of IRS-UAV with terrestrial IRS UAV networks while tackling eavesdroppers is investigated in \cite{Xiaowei_Pang}. The simulation results show that the IRS-UAV integration outperforms other benchmarks. } 

For both types of networks mentioned above, EE is a critical issue. Therefore, we briefly mention some relevant works below. 

\textit{ii) \textbf{EH technologies for  IRS networks:}} Because of massive numbers of devices and their energy consumption,  e.g., IoT, wireless power transfer (WPT)  is a fundamental solution\cite{Wang}.   WPT is the foundation of WPCN and SWIPT networks, where wireless devices harvest energy from radio frequency (RF) signals. In a WPCN, each device follows the harvest-and-then-transmit protocol  \cite{Wang}. However, in PS-based SWIPT networks, each device harvests energy and decodes information simultaneously \cite{zargari}. The IRS can be deployed in both WPCN and SWIPT networks to improve performance. For instance,  a WPCN utilizes an IRS with a  single antenna hybrid access point (AP) and two users capable of EH and information decoding (ID) \cite{Zheng1}. The throughput is then maximized by joint design of the power allocations, transmit time,
and IRS phase shifts. The total energy consumption and the EE of an IRS-assisted PS-based SWIPT network are optimized by jointly optimizing the PS ratios, active BS beamformers, and IRS phase shifts \cite{zargari2,zargari3}. In addition, the BS transmission power is minimized in \cite{zargari2021robust} for an  IRS-aided downlink multiple-input single-output (MISO) PS-based SWIPT network.  Reference \cite{zargari4} introduces a multi-objective optimization framework to balance the sum-rate maximization and the total harvested energy maximization.

Since the IRS consumes a small amount of power only,  we can power it with  SWIPT \cite{Gong, Wei,Ramezani}. This option makes the IRS a   self-sustainable device that can operate for a long time and a  hybrid energy/information relay. For such a self-sustainable IRS-empowered network, \cite{Gong} maximizes the SNR  by jointly optimizing the active and passive beamformers at the AP and the IRS, respectively. Reference \cite{Wei} maximizes the network sum rate by joint design of the AP beamformers and  IRS phase shifts/EH schedule. Moreover,  \cite{Ramezani}  maximizes the sum rate by jointly optimizing the resource allocation and IRS phase shifts based on time switching (TS) and PS architectures. Besides, \cite{Hu2021} studies the robust and secure   MISO downlink communication with a  self-sustainable IRS, which is powered by EH. The sum rate is maximized by jointly optimizing the AP beamformers and  IRS phase shifts/EH schedule while ensuring security. The result indicates that a significant tradeoff between the sum rate and self-sustainable IRS exists.

This paper is organized as follows. Section \ref{sec2} presents the network model, the channel model, and the transmission scheme. Section \ref {ProblemForm} formulates the key  problem.  In Section \ref{sec4}, we propose our approach and solution to obtain the optimization problem. In Section \ref{Complexity}, we analyze the complexity of the proposed algorithm. In Section \ref{sec5},  numerical results are presented. Finally, Section \ref{sec6} concludes the paper.

{\it Notations}: Vectors and matrices are expressed by
boldface lower case letters $\mathbf{a}$ and capital letters $\mathbf{A}$, respectively. For a square matrix $\mathbf{A}$, $\mathbf{A}^H$ and $\mathbf{A}^T$ are Hermitian conjugate transpose and transpose of a matrix, respectively. $\mathbf{I}_M$ denotes the $M$-by-$M$ identity matrix. $\text{diag}(\cdot)$ is the diagonalization operation. The Euclidean norm of a complex vector and the absolute value of a complex scalar are denoted by $\|\cdot\|$ and $|\cdot|$, respectively. The distribution of a circularly symmetric complex Gaussian (CSCG) random vector with mean $\boldsymbol{\mu}$ and covariance matrix $\mathbf{C}$ is denoted by $\sim \mathcal{C}\mathcal{N}(\boldsymbol{\mu},\,\mathbf{C})$. 
$\nabla_{\mathbf{x}}$ and $\partial _{x}$ denote the gradient vector and the respective partial derivative with respect to $\mathbf{x}$, respectively. The expectation operator is denoted by $\mathbb{E}[\cdot]$. Besides, $\mathbb{C}^{M\times N}$ and ${\mathbb{R}^{M \times 1}}$ represent $M\times N$ dimensional complex matrices and $M\times 1$ dimensional real vectors, respectively. Finally, $\lfloor x \rceil$ denotes the nearest integer of x, and $\mathcal{O}$ expresses the big-O notation.

\section{network Model}\label{sec2}

In  Fig. 1, an IRS-empowered UAV network comprising a single-antenna BS, an IRS-UAV, and $K$ users indexed by $\mathcal{K}=\{1,... , K\}$ is considered.  The  IRS is equipped with an EH circuit to power its operations (Fig. \ref{IRS_diagram}) \cite{Gong, Wei, Ramezani}. In particular, the IRS-UAV establishes LoS links for the users to transmit information in the uplink to the BS. Without loss of generality, we consider a 3D Cartesian coordinates network to describe the locations of transceivers. The locations of the BS and user $k$  are denoted by $\mathbf{q}_{b}=[{q}^x_{b},{q}^y_{b},{q}^z_{b}]^T \in {\mathbb{R}^{3 \times 1}} $ and $\mathbf{q}_{k}=[{q}^x_{k},{q}^y_{k},{q}^z_{k} ]^T \in {\mathbb{R}^{3 \times 1}}$,  $\forall k \in \mathcal{K}$, respectively, which are assumed to be fixed. The location and velocity of the IRS-UAV for $0 < t < T$, where $T$ is the  total time slot, are represented by $\mathbf{q}_u(t) = [q^x_u(t),q^y_u(t),q^z_u(t)]^T  \in {\mathbb{R}^{3 \times 1}}$ and $\mathbf{v}_u(t) =[v^x_u(t),v^y_u(t),v^z_u(t)]^T  \in {\mathbb{R}^{3 \times 1}} $, respectively. Since the trajectory of the IRS-UAV regularly varies over time, the total time slot of the IRS-UAV operation is divided into $N$ discrete points i.e., $\delta_T={T}/{N}$ sufficiently small time slots in which the IRS-UAV position is assumed static \cite{Pan1,Kwan}. Accordingly, we have $\mathbf{q}_u[n] = [q^x_u[n],q^y_u[n],q^z_u[n]]^T$, where $H_{\min} \leq q^z_u[n]\leq H_{\max},~\forall {{n}}\in\mathcal{N}=\{1,2,..., N\}$ indicates the minimum and maximum altitude limitation of the IRS-UAV. 
\begin{remark}
Note that although the UAV can fly at higher altitudes, it approaches each user as close as possible to maximize EH. Unlike previous works \cite{Pan2,Pan1,mahmoud2021,Nnamani,Zhang3}, we optimize the UAV altitude, making EH at the IRS more realistic. Accordingly, the IRS maximum altitude does not significantly impact our optimization problem.
\end{remark}

\begin{figure}[t]
\centering
\includegraphics[width=3in]{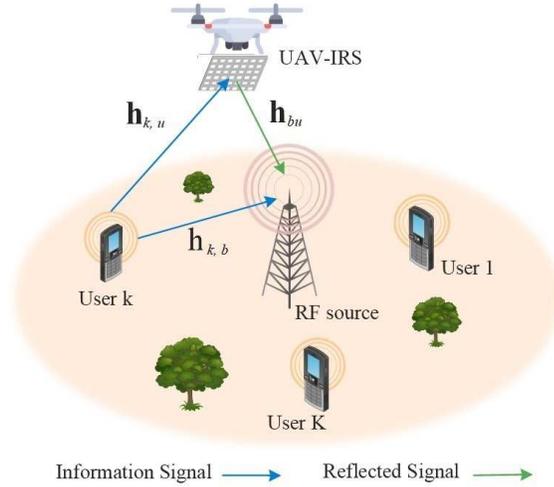}
\caption{A multiuser IRS-UAV network.}\label{network_model}
\end{figure}

Subsequently, the constraints on  the IRS-UAV position and velocity can be  expressed  as
\begin{align}\
& {\bf{q}}_u[1] = {{\bf{q}}_s},\quad {\bf{q}}_u[N + 1] = {{\bf{q}}_e},\\
& {{\bf{q}}_u[n + 1] - {\bf{q}}_u[n]}=  {\bf{v}}_u[n]\delta_T ,\;{\rm{ }}\forall {{n}}\in\tilde{\mathcal{N}}=\{2,3,..., N-1\},\nonumber\\
&\left\|{\bf{v}}_u[n + 1] - {\bf{v}}_u[n]\right\| \leq {a_{\max }}\delta_T,\;{\rm{ }}\forall {{n}}\in\tilde{\mathcal{N}},\nonumber\\& \left\|{\bf{v}}_u[n] \right\| \leq V_{\max},\;\forall {{n}}\in\mathcal{N},\nonumber
\end{align}
respectively, where $V_{\text{max}}$,  $a_{\max}$, ${{\bf{q}}_s}$, and ${\bf{q}}_{e}$   denote the maximum flying speed,  the maximum flight
acceleration, and the first and final positions of the IRS-UAV, respectively. 

\subsection{Channel Model}
This network model presupposes the availability of all channel state information (CSI)\footnote{The IRS network may acquire CSI in two ways, depending on whether the IRS reflecting elements are equipped with receive RF chains or not. If the IRS has RF chains, traditional techniques can be used to estimate the user-IRS and IRS-BS channel links. Otherwise, uplink pilots and  IRS reflection patterns can be designed to estimate the channel links \cite{Zhang}.% \cite{Liu,Zhang}.
}. We assume quasi-static block fading channels where the channels remain constant in each block and vary over blocks. Since the duration of each fading block is typically much smaller than $\delta_T$, the number of fading blocks at time slot $n\gg 1 $  in practice \cite{Zhan2018}. The  baseband equivalent channel links at time slot $n$ from the BS to the IRS-UAV, user $k$ to the IRS-UAV, and user $k$ to the BS are given by  ${\mathbf{h}_{bu}}[n]\in \mathbb{C}^{M \times 1}$, ${\mathbf{h}}_{k,u}[n]\in \mathbb{C}^{M\times 1}$, and ${ h}_{k,b}[n]\in \mathbb{C}^{1\times 1}$, respectively.

In addition, we consider a $M_x\times M_y$ reflection units at IRS where $M_x$ and $M_y$ denote the number of reflecting elements along the $x$-axis and $y$-axis,~respectively. In contrast to \cite{Pan2} where a uniform linear array (ULA) is considered at the IRS, we study a more general case and assume that reflection units at the IRS are spanned as a uniform planar array (UPA). The reflection coefficients matrix of IRS is given by ${\mathbf{\Theta}}[n]=\text{diag}(\boldsymbol{\theta}[n])$, where 
${ 	\boldsymbol{\theta}[n]\!=\!\! \left[\! {\rho_{1,1}[n]} e^{j\beta_{1,1}[n]}, \!...,\! {\rho_{M_x,M_y}[n]} e^{j\beta_{M_x,M_y}[n]}\! \right]^T \! \in \!\mathbb{C}^{M_xM_y\!\times \!1}}.$ Specifically, ${\rho_{m_x,m_y}[n]} \in [0,1]$ and $\beta_{m_x,m_y}[n] \in (0,2\pi]$, $\forall m_x,m_y \in \mathcal{M}=\{1,...,M_xM_y\}$ are the reflection coefficient and phase shift of the $(m_x,m_y)$-th reflecting element at the IRS, respectively. 
\begin{figure}[t]
\centering
\includegraphics[width=3in]{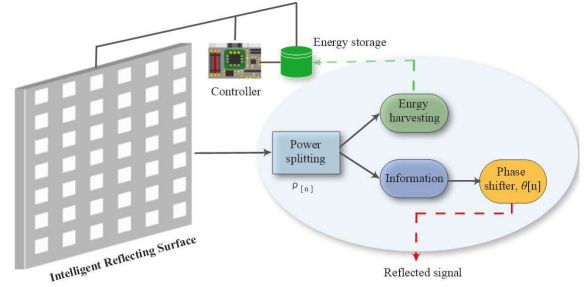}
\caption{Architecture of the IRS.}\label{IRS_diagram}
\end{figure}   	
For simplicity, we denote $m_x, m_y$ and $M_x, M_y$ as $m$ and $M$, respectively. All channel links follow the Rician model, and the channel link between the BS and the IRS-UAV as well as user $k$ and the IRS-UAV at time slot $n$ are given by 
\begin{align} 
&\mathbf{h}_{bu}[n]=\nonumber\\
&\sqrt{\beta_{bu}[n]}\left( \sqrt{\frac{K_{bu}}{K_{bu}+1}}\mathbf{h}_{bu}^{\mathrm{LoS}}[n]+\sqrt{\frac{1}{K_{bu}+1}}\mathbf{h}_{bu}^{\mathrm{NLoS}}[n]\right),\\
&\mathbf{h}_{k,u}[n]=\nonumber\\&\sqrt{\beta_{k,u}[n]}\left( \sqrt{\frac{K_{k,u}}{K_{k,u}+1}}\mathbf{h}_{k,u}^{\mathrm{LoS}}[n]+\sqrt{\frac{1}{K_{k,u}+1}}\mathbf{h}_{k,u}^{\mathrm{NLoS}}[n]\right),
\end{align}
respectively, where  $K_{bu}$ and $K_{k,u}$ represent Rician factors. In addition, we have $\beta_{bu}[n]=\frac{\beta_{0}}{{d_{bu}^{\alpha_{bu}}[n]}}$ and $\beta_{k,u}[n]=\frac{\beta_{0}}{{d^{\alpha_{k,u}}_{k,u}[n]}}$, where $\alpha_{bu}$ and $\alpha_{k,u}$ denote the path loss exponents and $\beta_{0}$ indicates the channel power at the reference distance of one meter. Furthermore, $d_{bu}[n]$ and $d_{k,u}[n]$ are the distance between the IRS-UAV and the BS-user $k$ and the IRS-UAV, respectively, which are given by $d_{bu}[n]={\|\mathbf{q}_b[n]-\mathbf{q}_{u}[n]\|^2}$ and $d_{k,u}[n]={\|\mathbf{q}_{k}[n]-\mathbf{q}_u[n]\|^2}$, respectively\footnote{Due to significant path loss and reflection loss, it is assumed that the power of the reflected signals by the IRS two or more times is negligible and hence ignored \cite{Kwan}.}. The elements of $\mathbf{h}_{bu}^{\mathrm{NLoS}}$ and $\mathbf{h}_{k,u}^{\mathrm{NLoS}}$ are assumed to be independent and identically distributed with zero mean and unit variance, i.e., $\mathbf{h}_{k,u}^{\mathrm{NLoS}},\mathbf{h}_{bu}^{\mathrm{NLoS}}\sim \mathcal{C}\mathcal{N}(\mathbf{0},\,\mathbf{I}_{M_xM_y})$. As shown in Fig. \ref{tetha}, we have
\begin{align}
\mathbf{h}_{bu}^{\mathrm{LoS}}[n]&\!=\!\left[\! 1,{e^{ \!-\!j \varrho{\sin {\theta}[n]\cos {\varphi}[n]}}},...,{e^{ \!-\! j\varrho({M_x} - 1){\sin {\theta}[n]\cos {\varphi}[n]}}}\right]^T  \nonumber\\ & \otimes\! \left[ \!1,{e^{ \!-\! j\varrho{\sin {\theta}[n]\sin {\varphi}[n]}}},...,{e^{ \!-\!j \varrho({M_y} - 1){\sin {\theta}[n]\sin {\varphi}[n]}}}\right]^T\!\!\!,\\
\mathbf{h}_{k,u}^{\mathrm{LoS}}[n]&\!=\nonumber\\
&\!\!\!\!\left[ 1,{e^{ \!-\!j \varrho{\sin {\omega _k}[n]\cos {\zeta _k}[n]}}},...,{e^{\! - \!j\varrho({M_x} \!-\! 1){\sin {\omega _k}[n]\cos {\zeta _k}[n]}}}\right]^T  \nonumber\\ &\!\!\!\!\!\!\!\!\!\!\otimes \!\left[ 1,{e^{ \!-\!j \varrho{\sin {\omega _k}[n]\sin {\zeta _k}[n]}}},...,{e^{ \!-\! j\varrho({M_y} \!-\! 1){\sin {\omega _k}[n]\sin {\zeta _k}[n]}}}\right]^T\!\!\!.
\end{align}
Specifically, $\varrho=\frac{2\pi f_cd}{c}$ where $d$ is the distance between two adjacent reflecting elements at the IRS and $f_c$ is the carrier frequency. $\theta[n]$ and $\varphi[n]$ are the vertical and horizontal angle of departure (AoD) between the IRS-UAV and the BS at time slot $n$, respectively. Similarly, $\omega_k[n]$ and $\phi_k[n]$ are the vertical and horizontal AoD between user $k$ and the IRS-UAV at time slot $n$, respectively. In particular, the AoDs/AoAs are given by 	
\begin{figure}[t]
\centering
\includegraphics[width=3.5in]{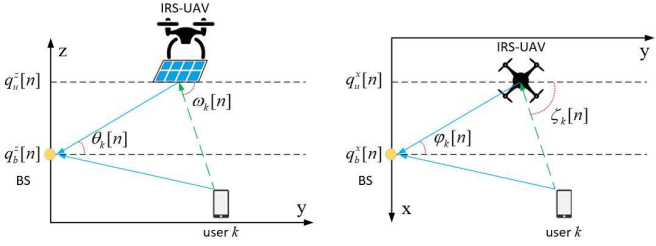}
\caption{The vertical and horizontal AoDs/AoAs between the IRS-UAV, the BS, and user $k$ are presented on the left-hand side and the right-hand side, respectively.  }\label{tetha}
\end{figure} 
\begin{align}
&\sin\theta[n]=\left( \frac{{q}^z_{b}-q^z_u[n]}{d_{bu}[n]}\right),\quad\sin\omega _k[n]=\left( \frac{q^z_u[n]}{d_{k,u}[n]}\right) ,\\
&\sin \varphi[n]=\left( \frac{{q^x_u[n]- {q}^x_{b}}}{{\sqrt {{{(q^x_u[n]- {q}^x_{b})}^2} + {{(q^y_u[n]- {q}^y_{b})}^2}} }}\right),\\
&\cos \varphi[n]=\left( \frac{{q^y_u[n]- {q}^y_{b}}}{{\sqrt {{{(q^x_u[n]- {q}^x_{b})}^2} + {{(q^y_u[n]- {q}^y_{b})}^2}} }}\right),\\
&\sin \zeta_k[n]=\left( \frac{ {q}^x_{k}-{q^x_u[n]}}{{\sqrt {{{(q^x_u[n]- {q}^x_{k})}^2} + {{(q^y_u[n]- {q}^y_{k})}^2}} }}\right),\\
&\cos \zeta_k[n]=\left( \frac{{{q}^y_{k}-q^y_u[n]}}{{\sqrt {{{(q^x_u[n]-{q}^x_{k})}^2} + {{(q^y_u[n]- {q}^y_{k})}^2}} }}\right).
\end{align}
Moreover, the channel link between the BS and user $k$ at time slot $n$ is given by
\begin{align}
&{h}_{k,b}[n]\!=\!\sqrt{\beta_{k,b}}\!\!\left(\! \sqrt{\frac{K_{k,b}}{K_{k,b}+1}}{h}_{k,b}^{\mathrm{LoS}}[n]\!+\!\sqrt{\frac{1}{K_{k,b}+1}}{h}_{k,b}^{\mathrm{NLoS}}[n]\right)\!\!,
\end{align}
where $K_{k,b}$ is the Rician factor, ${h}_{k,b}^{\mathrm{LoS}}[n]=e^{-j\frac{2\pi f_c d_{k,b}[n]}{c}}$, and ${h}_{k,b}^{\mathrm{NLoS}}[n]\sim \mathcal{C}\mathcal{N}({0},\,{1})$. In addition, we have $\beta_{k,b}=\frac{\beta_{0}}{d_{k,b}^{\alpha_{k,b}}}$, where $\alpha_{k,b}$ denotes the path loss exponent and  $d_{k,b}$ is the distance between  user $k$ and the BS. 
\subsection{Transmission Scheme}
The received signal at the $m$-th reflecting element of the IRS-UAV at time slot $n$ for all $ m\in \mathcal{M}$, $n\in \mathcal{N}$, and  $k\in \mathcal{K}$ can be expressed as
\begin{equation}
 y^{(m)}_{u,k}[n]=\sqrt{P_k} h^{(m)}_{k,u}[n]x_k[n],
\end{equation}
where $P_k$ and $x_k$ denote the transmit power and information of user $k$, respectively. Based on the received signal, the reflection coefficients of reflecting elements\footnote{Positive-intrinsic-negative (PIN) diodes, field-effect transistors (FET), micro-electromechanical network (MEMS) switches, and variable resistor loads can be employed for tuning the reflection coefficients of reflecting elements at the IRS \cite{Gong1,Ramezani}.} at the IRS can be adjusted in such a way as to reflect and harvest energy simultaneously. To simplify the circuit design of the IRS, we assume  that all reflecting elements have the same reflection coefficient, i.e., $\bar{\rho}[n]\!=\!\rho_{1}[n]\!=\!\cdots\!=\!\rho_{M}[n]$, which is a practical assumption \cite{Ramezani,Gong}. Accordingly, the linear EH model of all reflecting elements at time slot $n$ is given by $ E^h_{u,k}[n]=\eta s_k[n]P_k(1-\bar{\rho}^2[n])\|{ \mathbf{ h}}_{k,u}[n]\|^2,$ where $\eta\in (0,1]$ is the energy conversion efficiency. 

The EH and signal reflection parts at the IRS are managed by its microcontroller. The  power splitter is a simple passive device \cite{keysight_2021}, whose  basic schematic is shown in Fig. \ref{ps_stru}  \cite{in3otd}. The Tr1 transformer is used to bring the 50 $\Omega$ input impedance down to the 25 $\Omega$ needed to feed the second transformer, Tr2. This latter is doing the actual power splitting feeding the two outputs. The resistor between the outputs is needed for isolation \cite{in3otd}. With such a circuit, the power split ratio can be readily controlled.  The received signal at the IRS  can thus be split into two streams with one being reflected to the user and the remaining one fed into the EH circuit.  Besides, $s_k[n]$ denotes the user scheduling indicator at time slot $n$ which is represented as
 \begin{figure}[t]
\centering
\includegraphics[width=2.7in]{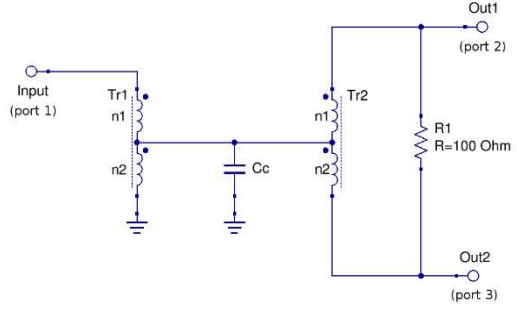}
\caption{The power splitter basic schematic comprises transformers  (Tr), resistors (R), and a capacitor (Cc) \cite{in3otd}. }\label{ps_stru}
\end{figure} 
\begin{equation*}
s_k[n]=
\begin{cases}
1, & \text{if user $k$ is scheduled at time slot $n$},  \\
0, & \text{otherwise}.
\end{cases}
\end{equation*}
However, the linear EH model is not practical since it cannot capture the effect of the non-linearity behavior of circuits~\cite{Wang2020}. Hence, we resort to a parametric non-linear EH model based on the sigmoidal function \cite{zargari3,Nonlinear}. Then, the total harvested energy at the IRS can be written as
\begin{subequations}\label{EH}\begin{align}
&\bar{ E}^h_{u,k}[n] = \frac{{\Upsilon_k[n] - {M_k}{\Omega_k}}}{{1 - {\Omega_k}}}, ~{\Omega_k} = \frac{1}{{1 + e^{c_k \nu_k}}}{,~\forall n, k}, \nonumber\\& \Upsilon_k[n] = \frac{{{M_k}}}{{1 + \exp ( - {c_k}({ E^h_{u,k}[n]} - {\nu_k}))}}{,~\forall n, k,}
\end{align}
\end{subequations}
where $\Upsilon_k$ is the traditional logistic function and ${\Omega_k}$ is a constant to guarantee a zero input/output response. In addition, $c_k$ and $\nu_k$ are constants parameters related to the circuit characteristics and $M_k$ is a constant denoting the maximum harvested energy at the IRS. Note that all parameters can be obtained by a curve fitting tool \cite{Nonlinear}. 

\begin{remark} Assume that the energy consumption of each reflecting element at the IRS is denoted by $p_{m}$, so the overall energy consumption of the IRS can be written as $E_\text{min}=M_xM_yp_{m}$. To ensure that the harvested energy is able to power the IRS, we need to have the following constraint: $\bar{ E}^h_{u,k}[n]\geq E_\text{min}$,  a non-convex and non-linear constraint. However, it can be rewritten as $E^h_{u,k}[n]\geq {\nu_k} - \frac{1}{{{c_k}}}\ln (\frac{{{M_k} - E_\text{min}}}{{E_\text{min}}})$ \cite{zargari3}, which thus makes optimization problem more tractable. 
\end{remark}

The received signal at the BS which is the superposition of the direct and reflected signal can be written as
\begin{align}
y_{k,b}[n]&=\sqrt{P_k}\left({h}_{k,b}[n]+\bar{\rho}[n]\left( \mathbf{h}_{k,u}[n]\right) ^T {\mathbf{\Theta}}[n]\mathbf{h}_{bu}[n]\right)  x_k[n]\nonumber\\&+z_k[n],~\forall n,k,
\end{align}
where $z_k{[n]}\sim \mathcal{N}(0,\,\sigma^{2})\,$ is the received noise with variance %\footnote{It should be noted that the noise power is assumed to be the same for all the users.} 
$\sigma^{2}$. The achievable data rate at time slot $n$ can be expressed as
\begin{equation}
R_{k}[n]=s_k[n]\log_{2}\left( 1+\text{SNR}_k[n]\right),~\forall  k,
\end{equation}
where signal-to-noise ratio (SNR) is given by
\begin{equation}\text{SNR}_k[n]={ \frac{P_k}{{\sigma^{2}}}\left|  {h}_{k,b}[n]+\bar{\rho}[n]\mathbf{h}^H_{k,u}[n]{\mathbf{\Theta}}[n]\mathbf{h}_{bu}[n]\right| ^2},~\forall  k.	
\end{equation}
\begin{remark}
Note that because of the variation of the IRS-UAV locations, the distribution of the channel remains constant in each time slot but changes over different time slots. Consequently, over each time slot, the transmission rate by the scheduled user can be adapted to the IRS-UAV location. When the trajectory, velocity, user scheduling, and transmission rate are obtained, the IRS-UAV, with the help of the BS, schedules the corresponding users and notifies each user of the optimized transmission rate over time slots by utilizing the downlink control links. 
\end{remark}
In the following theorem, we derive the average achievable data rate for the formulated optimization since the involving channel links vary relatively fast due to the small-scale fading phenomenon.
\begin{theorem}
The upper bound of the average achievable data rate for each user is given as follows: 
\begin{align}	 \mathbb{E}\{R_{k}[n]\}&\leq s_k[n]\log_{2}\left( 1+\frac{P_k}{\sigma^2}\left( |b_k|^2+  \frac{ \beta_{k,b}}{K_{k,b}+1}\right.\right.\nonumber\\  & \hspace{-2em}\left.\left.+\frac{\bar{\rho}^2[n](K_{k,u}+K_{bu}+1)M_xM_y\beta_{bu}[n]\beta_{k,u}[n]}{(K_{k,u}+1)(K_{bu}+1)} \right) \right),\label{200}
\end{align}
where
\begin{align}
b_k&=\sqrt{\frac{K_{k,b}\beta_{k,b}}{K_{k,b}+1}}h_{k,b}^{\mathrm{LoS}}[n]+\sqrt{\frac{K_{k,u}K_{bu}\beta_{bu}[n]\beta_{k,u}[n]}{(K_{k,u}+1)(K_{bu}+1)}}\nonumber\\&\times\bar{\rho}[n](\mathbf{h}_{k,u}^{\mathrm{LoS}}[n])^{H}{\mathbf{\Theta}}[n]\mathbf{h}_{bu}^{\mathrm{LoS}}[n].
\end{align}

\end{theorem}

\begin{proof}
See Appendix \ref{sec:Theorem1}.
\end{proof}

As can be observed, this data rate is a function of the large-scale fading coefficients, LoS channel components, and IRS phase shifts. We note that the above approximation will be tight if the SNR is sufficiently high  \cite{Pan2,Han2019}. Compared to previous works \cite{JiaoFZZ20,Zina_Mohamed,Xiaowei_Pang}, where the UAV carries the IRS, we average out the small-scale fading. However, we apply previously recognized techniques to obtain the upper bound of the average achievable data to a new problem and system model  \cite{Pan2,Han2019}.  Subsequently, based on Theorem 1, we obtain the average harvested energy at the IRS which is given by
\begin{equation}
\mathbb{E}\{E^h_{u,k}[n]\}\leq s_k[n]P_k(1-\bar{\rho}^2[n])M_xM_y\beta_{bu}[n].
\end{equation}

\section{Problem Formulation}\label{ProblemForm}
In this section, we formulate an optimization problem to jointly optimize the user scheduling $\boldsymbol{s}=\{s_k[n],\forall k,n\}$, UAV trajectory $\boldsymbol{q}_u=\{\mathbf{q}_u[n],\forall n\}$, UAV velocity   $\boldsymbol{v}_u=\{\mathbf{v}_u[n],\forall n\}$, IRS phase shifts $\boldsymbol{\Phi}=\{{\mathbf{\Theta}}[n],\forall n\}$, and IRS reflection coefficient $\boldsymbol{\rho}=\{\bar{\rho}[n],\forall n\}$ to minimize
the maximum energy consumption of all users. { By defining $E_k \triangleq \delta_T P_k $ in (Joule)}, and $R_{\min,k}\triangleq\frac{S_k}{B\delta_T}$
in (bits/Hz), where $B$ denotes the channel bandwidth in (Hz); then the problem can be mathematically formulated as below:
\begin{align}\label{M_U}
\!\text{(P1)}:~& 	\underset{\substack{ \boldsymbol{s},\boldsymbol{q}_u,\boldsymbol{v}_u\\\boldsymbol{\Phi},\boldsymbol{\rho},\kappa }} {\text{minimize}}~\kappa \nonumber \\
&\text{s.t.}~C_{1}:~\sum_{n=1}^{N}s_k[n]E_k \leq \kappa,~\forall k,\nonumber\\
&\quad~C_{2}:~\sum_{n=1}^{N}\mathbb{E}\{E^h_{u,k}[n]\}\geq {\nu_k} \!- \!\frac{1}{{{c_k}}}\ln (\frac{{{M_k} \!-\! E_\text{min}}}{{E_\text{min}}}),~\forall k,  \nonumber\\
&\quad~C_{3}:~\sum_{n=1}^{N}\mathbb{E}\{R_{k}[n]\}\geq R_{\min,k},~\forall k,\nonumber\\
&\quad~C_{4}:~\sum_{k=1}^{K}s_k[n]\leq 1,~\forall n,\nonumber\\
&\quad~C_{5}:~s_k[n]\in\{0,1\},~\forall k,n,  \nonumber\\
&\quad~C_{6}:~0\leq\beta_{m}[n]\leq 2\pi,~\forall m,n,\nonumber\\
&\quad~C_{7}:~0\leq\bar{\rho}[n]\leq 1,~\forall n,\nonumber\\
&\quad~C_{8}:~{\bf{q}}_u[1] = {{\bf{q}}_s},\quad {\bf{q}}_u[N + 1] = {{\bf{q}}_e},\nonumber\\
&\quad~C_{9}:~{{\bf{q}}_u[n + 1] - {\bf{q}}_u[n]}=  {\bf{v}}_u[n]\delta_T ,\;{\rm{ }}\forall {{n}},\nonumber\\
&\quad~C_{10}:~\left\|{\bf{v}}_u[n + 1] - {\bf{v}}_u[n]\right\| \leq {a_{\max }}\delta_T,\;{\rm{ }}\forall {{n}},\nonumber\\
&\quad~C_{11}:~\left\|{\bf{v}}_u[n] \right\| \leq V_{\max},\;\forall {{n}},
\end{align}
where  $C_{1}$ ensures that the energy consumption of each user does not exceed $\kappa$, where $\kappa$ is the slack variable indicating the maximum energy to be minimized. $C_{2}$ and $C_{3}$ ensure the minimum required data rate and harvested energy, respectively. $C_{4}$ and $C_{5}$ are user scheduling constraints. $C_{6}$ and $C_{7}$ are the phase shifts and reflection coefficient constraints at the IRS, respectively. $C_{8}$ indicates the initial and final locations of the IRS-UAV. $C_{9}$ denotes the relation between the UAV’s trajectory and its flight velocity.  $C_{10}$ and $C_{11}$
are the maximum flight velocity and the maximum flight acceleration, respectively.

\begin{figure}[t]
\centering
\includegraphics[width=3.3in]{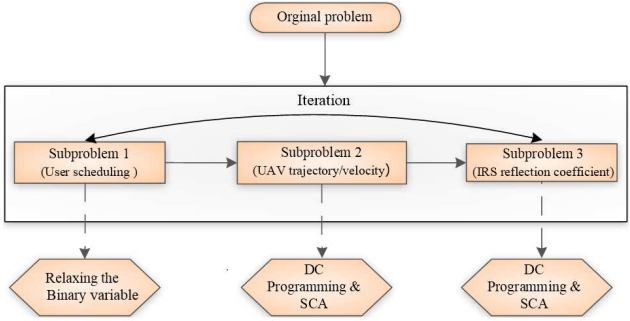}
\caption{A flow chart of the proposed AO algorithm.}\label{Algo_diagram}

\end{figure}

\section{Proposed Solution}\label{sec4}

The optimization problem (P1) is a non-convex and mixed-integer problem due to the non-linear multiplication of optimization variables in $C_{1}-C_{3}$ and $C_{5}$, which renders the solution challenging. Such problems are, in general, NP-hard and mathematically intractable. Although an exhaustive search in the feasible set of the variables may provide a globally optimal solution, it has exponential complexity. To enable real-time computations and reduce complexity, we propose an AO algorithm. 

The key idea of the AO algorithm is as follows. To solve, $\min_{x} f(x)$, where $x\in \mathbb R^{s},$   we partition  $x$ into  $m >1 $  blocks as $ x=(x_1,x_2,\ldots, x_m)^T$, where $x_k\in \mathbb R^{s_k}$  and $\sum_{k=1}^m  s_k =s.$ The minimization is then performed   over $x_1$ while $\{x_k| k\neq 1\}$   are kept constant over their previous values; next, it is performed   over $x_2$ while $\{x_k |k\neq 2\} $   are kept constant over their previous values.  This process thus cycles until convergence.   The  AO algorithm yields a  locally optimum solution. 

To apply the AO algorithm, we take two steps. First, we derive a closed-form solution for the IRS  phase shifts. Second, we divide the main problem into three sub-problems. In the first one, we relax the binary constraint of user scheduling into a continuous one and then solve it by CVX \cite{Boyd,Zhan2018}. In the second and third ones, we apply the DC programming and successive convex approximation (SCA) techniques to get suboptimal solutions for the UAV trajectory/velocity and IRS reflection coefficient, respectively.  Fig. \ref{Algo_diagram} shows the essential steps for the solution of (P1). The terrestrial BS executes the proposed iterative algorithm. Thus, the energy consumption of this execution is minuscule compared to the total power consumption of the BS, which must execute power-hungry operations, including RF processing, analog-to-digital conversions, and many others.  

\subsection{Optimizing Phase Shits:}
We next derive a closed-form solution for IRS phase shifts while other optimization variables remain fixed.  To do this, we maximize the primary data rate expression  (\ref{200}). As a result, we first rewrite $|b_k|$ given in (\ref{thetaa}), where
\begin{table*}	 
\centering
\begin{minipage}{1\textwidth}  
	\begin{align}				\displaystyle &|b_k|=\nonumber\\&\bigg|\underbrace{\sqrt{\frac{K_{k,b}\beta_{k,b}}{K_{k,b}+1}}e^{-j\frac{2\pi f_c d_{k,b}[n]}{c}}}_{F_1}+\underbrace{\sqrt{\frac{K_{k,u}K_{bu}\beta_{bu}[n]\beta_{k,u}[n]\bar{\rho}^2[n]}{(K_{k,u}+1)(K_{bu}+1)}}
		e^{-j\frac{2\pi f_c (d_{bu}[n]-d_{k,u}[n])}{c}}\times	\sum_{m_i=1}^{M_i}\sum_{m_j=1}^{M_j}e^{-j\frac{2\pi f_cd}{c}\vartheta_{m_i,m_j}[n]+j\beta_{m_i,m_j}[n]}}_{F_2}\bigg|\label{thetaa}
	\end{align}			\hrule
\end{minipage}
\end{table*}
\begin{align}
\vartheta_{m_i,m_j}[n]&=(m_i-1)\sin {\theta}[n]\cos {\varphi}[n]\nonumber\\&+(m_j-1)\sin {\theta}[n]\sin {\varphi}[n]\nonumber\\&+(m_i-1)\sin {\omega _k}[n]\cos {\zeta _k}[n]\nonumber\\&+(m_j-1)\sin {\omega _k}[n]\sin {\zeta _k}[n].
\end{align}
Then by using the triangle inequality, the upper bound of (\ref{thetaa}) can be obtained as follows: $|F_1+F_2|\leq |F_1|+|F_2|$ \cite{Boyd}, where the inequality holds if and only if the phase of first and second terms hold with equality which thus yields the optimal phase shifts\footnote{The phase shifts in closed-form expressions facilitate the control, synchronization, and channel estimation requirements as the channel link only depend on the UAV position. In this way, the phase control and required signaling overhead of the CSI acquisition can be reduced significantly.} as follows:
\begin{align}\label{phase}
\beta^\text{opt}_{m,k}[n]&=\frac{2\pi f_c(d_{bu}[n]-d_{k,u}[n])}{c}+ \frac{2\pi f_cd}{c}\vartheta_{m}[n]\nonumber\\&-\frac{2\pi f_c d_{k,b}[n]}{c},~\forall m.
\end{align}
Based on (\ref{phase}),  the phase shifts at the IRS should be tuned such that the direct path aligns with the reflected one. In this way, coherent combining improves the signal strength of each user.

\subsection{Optimizing  User Scheduling:}
We next derive the user scheduling for fixed UAV trajectory/velocity and IRS reflection coefficient. Consequently, (P1) can be recast as follows:
\begin{align}
\text{(P2)}:~& \underset{\boldsymbol{s},\kappa} {\text{minimize}}\:~\kappa\nonumber \\
&\text{s.t.}~C_{1}-C_{5}.
\end{align}
To derive a computationally efficient resource allocation and handle the non-convexity of (P2), we relax the binary constraint of the user scheduling into a continuous one, i.e., $0 \leq s_k[n] \leq  1$. Accordingly, the new optimization problem can be written as
\begin{align}
\text{(P3)}:~& \underset{\boldsymbol{s},\kappa} {\text{minimize}}\:~\kappa\nonumber \\
&\text{s.t.}~C_{1}-C_{4},\nonumber\\
&\quad~C_{5}:~0 \leq s_k[n]\leq 1,~\forall k,n.
\end{align}
Since (P3) is a standard linear program (LP), it can be efficiently solved by CXV \cite{Boyd}. The binary user scheduling variables can be recovered as follows. 
Indeed, there are $LN$ fading blocks in total with the time horizon $T$. If the obtained solution of the relaxed problem is not binary, $N_k[n] = \lfloor Ls_k[n]  \rceil $ fading blocks can be allocated to user $k$ in any time slot $n$ \cite{Zhan2018}.

\subsection{Optimizing UAV Trajectory and Velocity:}
We next optimize the UAV trajectory/velocity for fixed user scheduling and IRS reflection coefficient. Indeed, the UAV trajectory/velocity are optimized to maximize the weighted minimum of the data rate of all users, where the weight is inversely proportional to $R_{\min,k}$. In particular, the  optimization problem has the following structure:
\begin{align}\label{P4}
\text{(P4)}:~& \underset{\boldsymbol{q}_u,\boldsymbol{v}_u,\chi} {\text{maximize}}\:\:\chi\nonumber \\
&\text{s.t.}~C_{2},~C_{8}-C_{11},\nonumber\\ 
&\quad~C_{3}:~\frac{1}{R_{\min,k}}\sum_{n=1}^{N}\mathbb{E}\{R_{k}[n]\}\geq \chi,~\forall k.
\end{align}
Based on (\ref{phase}), the term $|b_k|^2$ in constraint $C_3$ can be stated as
\begin{align}\label{b}
|b_k|^2&=\frac{K_{k,b}\beta_{k,b}}{K_{k,b}+1}+\frac{M^2_xM^2_yK_{k,u}K_{bu}\beta_{bu}[n]\beta_{k,u}[n]\bar{\rho}^2[n]}{(K_{k,u}+1)(K_{bu}+1)}\nonumber\\&+2M_xM_y\sqrt{\frac{K_{k,b}K_{k,u}K_{bu}\beta_{k,b}\beta_{bu}[n]\beta_{k,u}[n]\bar{\rho}^2[n]}{(K_{k,b}+1)(K_{k,u}+1)(K_{bu}+1)}}.
\end{align}
To this end, the objective function in (P4) can be transformed into 
\begin{align} 
\mathbb{E}\{R_{k}[n]\}&\leq s_k[n]\log_{2}\bigg( \!\!1\! +\!\frac{P_k}{\sigma^{2}}\big(\! \left(  \psi_1[n]+\psi_2[n] \right)\beta_{bu}[n] \beta_{k,u}[n] 
\nonumber\\
& +\psi_3[n]\sqrt{\beta_{bu}[n]\beta_{k,u}[n]}+\psi_4\beta_{k,b}[n]\! \big)   \bigg)\overset{\Delta}{=} \Gamma_k[n], \label{29}
	\end{align} 
where 
\begin{align}	
&\psi_1[n]=\frac{K_{k,u}K_{bu}M_x^{2}M_y^{2}\bar{\rho}^2[n]}{(K_{k,u}+1)(K_{bu}+1)},\nonumber\\ &\psi_2[n]=\frac{\bar{\rho}^2[n](K_{k,u}+K_{bu}+1)M_xM_y}{(K_{k,u}+1)(K_{bu}+1)},\nonumber\\
&\psi_3[n]=2M_xM_y\sqrt{\frac{K_{k,b}K_{k,u}K_{bu}\beta_{k,b}\bar{\rho}^2[n]}{(K_{k,b}+1)(K_{k,u}+1)(K_{bu}+1)}},\nonumber\\
&\psi_4=\frac{K_{k,b} }{K_{k,b}+1}.
\end{align} 
Upon rearranging terms, problem (P4) can be rewritten as
\begin{align}\label{P5}
\text{(P5)}:~& \underset{\boldsymbol{q}_u,\boldsymbol{v}_u,\chi} {\text{maximize}}\:\:\chi \nonumber \\
&\text{s.t.}~C_{2},~C_{8}-C_{11},\nonumber\\ 
&\quad~C_{3}:~\frac{1}{R_{\min,k}}\sum_{n=1}^{N}\Gamma_k[n] \geq \chi,~\forall k.
\end{align}
Note that (P5) is not a convex optimization problem. To simplify the problem formulation, we propose $\beta_{bu}[n]=u[n]$ and $\beta_{k,u}[n]=r_k[n]$ as slack variables for trajectory planning. Thus, (\ref{29}) and constraint $C_2$ can be represented as
\begin{align}	
\!\!\!\!\bar{ \Gamma}_k[n]&\overset{\Delta}{=}  s_k[n]\log_{2}\bigg( 1+\frac{P_k}{\sigma^{2}}\big( \left( \psi_1[n]+\psi_2[n]\right)u[n] r_k[n]\nonumber\\
&\hspace{5em}+\psi_3[n]\sqrt{u[n]r_k[n]}+\psi_4\beta_{k,b}[n]\big) \bigg),\label{39}\\
\!\!\!\!\bar{C}_{2}\!\!:~& \!s_k[n]P_k(1\!-\!\bar{\rho}^2[n])M_xM_yu[n]\geq {\nu_k} \!-\! \frac{1}{{{c_k}}}\ln (\!\frac{{{M_k}\! -\! E_\text{min}}}{{E_\text{min}}}\!),\label{40}
\end{align} 
respectively. Therefore, (P5) can be reformulated as
\begin{align}\label{P6}
\!\!\!\!\text{(P6)}:~& 		\underset{\substack{ \boldsymbol{q}_u,\boldsymbol{v}_u,\chi\\u[n],r_k[n]}} {\text{maximize}}\:\:\chi\nonumber \\
&\text{s.t.}~\bar{C}_{2},~C_{8}-C_{11},\nonumber\\
&\quad~C_{3}:~\frac{1}{R_{\min,k}}\sum_{n=1}^{N}\bar{\Gamma}_k[n] \geq \chi,~\forall k,\nonumber\\
&\quad~C_{12}:~\beta_{bu}[n]\geq u[n],~ \beta_{k,u}[n]\geq r_k[n], \forall n, k.
\end{align}
By doing such transformation, constraint $C_{12}$ holds with the equality at the optimal solution, yielding that (P6) and (P5) are equivalent. This is because by increasing the value of $u[n]$ and $r_k[n]$, the value of the objective increases as well. However, problem (P6) is still non-convex due to  coupling between optimization variables, i.e., $\beta_{bu}[n]$ and $\beta_{k,u}[n]$. To overcome it,  we obtain a lower bound by using first-order Taylor expansion and the SCA technique. Consider a function $f(x)$ that depends on variable $x$, the first-order Taylor series at point $a$ is given by $
f(x)=f(a)+\partial _{x}f(x)\big|_{x=a}(x-a)$. By applying this, the lower-bound of constraint $C_{12}$ can be obtained as follows:
\begin{align}\label{key}
\beta_{k,u}[n]&\geq\nonumber\\&	\hspace{-5mm}\frac{\beta_{0}}{({\|\mathbf{q}^{(i)}_u[n]\!-\!\mathbf{q}_{k}[n]\|^2})^{\frac{\alpha_{k,b}}{2}}}-\frac{\alpha_{k,u}\beta_{0}}{2(\|\mathbf{q}^{(i)}_u[n]\!-\!\mathbf{q}_{k}[n]\|^2)^{\frac{{\alpha_{k,b}}}{2}+1}} \nonumber\\&\hspace{-5mm}\times({\|\mathbf{q}_u[n]\!-\!\mathbf{q}_{k}[n]\|^2}\!-\!{\|\mathbf{q}^{(i)}_u[n]\!-\!\mathbf{q}_{k}[n]\|^2})\nonumber\\
\beta_{bu}[n]&\geq\nonumber\\&\hspace{-5mm}\frac{\beta_{0}}{({\|\mathbf{q}^{(i)}_u[n]\!-\!\mathbf{q}_{b}[n]\|^2})^{\frac{\alpha_{bu}}{2}}}\!-\!\frac{\alpha_{bu}\beta_{0}}{2(\|\mathbf{q}^{(i)}_u[n]\!-\!\mathbf{q}_{b}[n]\|^2)^{\frac{{\alpha_{bu}}}{2}+1}} \nonumber\\&\hspace{-5mm}\times({\|\mathbf{q}_u[n]\!-\!\mathbf{q}_{b}[n]\|^2}\!-\!{\|\mathbf{q}^{(i)}_u[n]\!-\!\mathbf{q}_{b}[n]\|^2}).
\end{align}
To facilitate the solution design, we also introduce a new variable given as
\begin{equation}\label{36}
u[n]\geq\frac{t^2_k[n]}{r_k[n]},~\forall k, n.
\end{equation}
Consequently, $\bar{ \Gamma}_k[n]$ in (\ref{39}) is transformed into: 
\begin{align}			
\bar{ \Gamma}_k[n]&\overset{\Delta}{=}  s_k[n]\log_{2}\bigg( 1+\frac{P_k}{\sigma^{2}}\left( \big( \psi_1[n]+\psi_2[n]\right)t^2_k[n] \nonumber\\& \hspace{7em}+\psi_3[n]t_k[n]+\psi_4\beta_{k,b}[n]\big)\bigg).\label{41}
\end{align}
\begin{algorithm}[t]
\algsetup{linenosize=\scriptsize }
\caption{Successive Convex Approximation}
\begin{algorithmic}[1]\label{algorithm1}
\renewcommand{\algorithmicrequire}{\textbf{Input:}}
\renewcommand{\algorithmicensure}{\textbf{Output:}}
\REQUIRE Set the number of iterations ${t}$, initial point $\boldsymbol{q}_u^{(1)}$, and error tolerance $\epsilon_1 \ll 1$.
\STATE \textbf{repeat}\\
\STATE \quad For given $\{ \boldsymbol{s}^{(t)}, \boldsymbol{\rho}^{(t)}\}$, solve (P7) and store the\\ \quad intermediate solutions $\{{\boldsymbol{q}^{(t)}_u},\boldsymbol{v}^{(t)}_u\}$.\\
\STATE \quad Set ${t}={t}+1$;
\STATE   \textbf{until} $|\chi^{(t)}-\chi^{(t-1)}|\leq \epsilon_1$.
\STATE \textbf{Return} $\{\boldsymbol{q}^*_u,\boldsymbol{v}^*_u\}=\{{\boldsymbol{q}^{(t)}_u},\boldsymbol{v}^{(t)}_u\}$.
\end{algorithmic}
\end{algorithm}{Since ${\bar{\Gamma}}_k[n]$ is a differentiable  function with respect to $t_k[n]$, (\ref{41}) can be lower bounded as follows:}
\begin{align}\label{399}
{\bar{\Gamma}}_k[n]&\geq {\bar{\Gamma}}_k[n]\bigg|_{t^{(i\!-\!1)}_k[n]}\!\!\!\!+\!\left( \!\partial _{t_k[n]}{\bar{\Gamma}}_k[n]\bigg|_{t^{(i\!-\!1)}_k[n]}\right) (t_k[n]\!-\!t^{(i\!-\!1)}_k[n])\nonumber\\&\overset{\Delta}{=}{{\Psi}}_k[n].
\end{align}
Finally, the UAV trajectory/velocity optimization problem can be expressed as 
\begin{align}\label{P7}
\text{(P7)}:~& \underset{\substack{ u[n],r_k[n],t_k[n]\\\boldsymbol{q}_u,\boldsymbol{v}_u,\chi }} {\text{maximize}}\:\:\chi\nonumber \\
&\text{s.t.}~\bar{C}_{2},~C_{8}-C_{12},~(\text{\ref{36}})\nonumber\\
&\quad~C_{3}:~\frac{1}{R_{\min,k}}\sum_{n=1}^{N}{{\Psi}}_k[n] \geq \chi,~\forall k.
\end{align}
As a result, (P7) is a convex optimization problem and provides a lower bound for (P5). In particular, (P7) can be solved at iteration $i$ by employing a convex optimization solver, e.g., CVX, which is summarized in \textbf{Algorithm \ref{algorithm1}}. 
{\begin{remark}\label{reamrk4}
%	In problem (P4), a lower bound is derived for the average achievable data rate, and it is replaced in the constraint set of (P4) with this lower bound, based on which problem (P5) is obtained. Thus, 
The feasible solution to (P7) is always feasible to (P4) \cite{Zhan2018}.  On the other hand, let us denote  $F(\boldsymbol{s}^{(i)} , \boldsymbol{q}^{(i)}_u, \boldsymbol{\rho}^{(i)} )=\frac{1}{R_{\min,k}}\sum_{n=1}^{N}{{\Psi}}_k[n]$ at iteration $i$, it then follows that
\begin{equation}
1 \overset{\text{(a)}} \leq \underset{k} {\text{min}}\:\: \frac{1}{R_{\min,k}} F(\boldsymbol{s}^{(i)}, \boldsymbol{q}_u^{(i)}, \boldsymbol{\rho}^{(i)} )\overset{\text{(b)}} = \chi(\boldsymbol{s}^{(i)}, \boldsymbol{v}_u^{(i)},\boldsymbol{q}^{(i)}_u, \boldsymbol{\rho}^{(i)} ),
\end{equation}
where (a) holds since  $\{\boldsymbol{s}^{(i)} ,\kappa^{(i)} \}$  is a solution set of (P2) with given $\{\boldsymbol{q}^{(i)}_u,\boldsymbol{v}^{(i)}_u,\boldsymbol{\rho}^{(i)}\}$, thus it must satisfy $C_{3}$ in \eqref{M_U}. Besides, (b) is due to the definition of problem (P4). 
\end{remark}}
\begin{proposition}
The approximation (\ref{399}) produces a tight lower bound of $\Gamma_k[n]$, leading to a sequence of improved solutions for (P4).
\end{proposition}
\begin{proof}
See Appendix \ref{sec:Proposition1}.
\end{proof}

\subsection{Optimizing IRS Reflection Coefficient:}	
For any given user scheduling and UAV trajectory/velocity, we  maximize the weighted minimum of the harvested energy, where the weight is inversely proportional to $e_k=\left( {\nu_k} - \frac{1}{{{c_k}}}\ln (\frac{{{M_k} - E_\text{min}}}{{E_\text{min}}})\right) $. Thus, the reflection coefficient optimization problem can be written as
\begin{align}\label{rho_2}
\text{(P8)}:~& \underset{\boldsymbol{\rho},\Pi} {\text{maximize}}\:\Pi \nonumber \\
&\text{s.t.}~C_{3},~C_{7},\nonumber\\
&\quad~C_{2}:~\frac{\sum_{n=1}^{N}\mathbb{E}\{E^h_{u,k}[n]\}}{e_k}\geq \Pi,~\forall k,
\end{align} 
Upon rearranging terms, problem (P8) can be rewritten as
\begin{align}\label{rho3}
\text{(P9)}:~& \underset{\boldsymbol{\rho},\Pi} {\text{maximize}}\:\Pi \nonumber \\
&\text{s.t.}~C_{3},~C_{7},\nonumber\\
&\quad~C_{2}:~\frac{\sum_{n=1}^{N}s_k[n]P_k(1-\bar{\rho}^2[n])M_xM_y\beta_{bu}[n]}{e_k}\geq \Pi.
\end{align} 
However, (P9) is non-convex with respect to $\boldsymbol{\rho}$. To handle non-convexity, we resort to the SCA technique to approximate constraint $C_3$, which leads to the following lower bounded:
\begin{align}
{{\bar{\Gamma}}}_k[n]\!&\geq\! {{\bar{\Gamma}}}_k[n]\big|_{\rho^{(i\!-\!1)}_k[n]}\!\!+\!\!\left(\!\partial_{\rho_k[n]}{{\bar{\Gamma}}}_k[n]\big|_{\rho^{(i\!-\!1)}_k[n]}\!\right) \!(\!\rho_k[n]\!-\!\rho^{(i\!-\!1)}_k[n]).
\end{align}
It is worth mentioning that similar to Remark \ref{reamrk4}, it can be proven that the feasible solution to (P9) is always feasible to (P8). In this way, (P9) becomes a convex optimization problem that can be efficiently solved by CVX \cite{Boyd}. The SCA algorithm for solving (P9) is similar to \textbf{Algorithm \ref{algorithm1}} and is omitted here for brevity. Similar to Proposition 1, the SCA algorithm for solving (P9) monotonically increases the objective function of (P9) at each iteration and finally converges. The whole steps for solving (P1) are provided in \textbf{Algorithm \ref{algorithm2}}. In particular, the resulting objective values of (P7) and (P9) are non-decreasing over the iterations, which ensures the convergence of \textbf{Algorithm \ref{algorithm2}}.
\begin{algorithm}[t]
\caption{Overall Algorithm for Solving Problem (P1)}
\begin{algorithmic}[1]\label{algorithm2}
	\renewcommand{\algorithmicrequire}{\textbf{Input:}}
	\renewcommand{\algorithmicensure}{\textbf{Output:}}
	\REQUIRE Set the number of iterations ${i}$ and error tolerance $\epsilon_2 \ll 1$.
	\STATE \textbf{repeat}\\
	\STATE \quad Solve problem (P4) to obtain the user scheduling $\boldsymbol{s}^{(i)}$. 
	\STATE \quad  Use \textbf{Algorithm \ref{algorithm1}} to optimize the UAV \\ \quad trajectory/velocity and obtain $\mathbf{q}^{(i)}_u$ and  $\mathbf{v}^{(i)}_u$ .
	\STATE \quad Obtain the optimal IRS reflection coefficient, $\boldsymbol{\rho}^{(i)}$\\ \quad according to (P9).
	\STATE \quad Set $i=i\!+\!1$;
	\STATE   \textbf{until}  $|\kappa^{(i)}-\kappa^{(i-1)}|\leq \epsilon_2$.
	\STATE \textbf{Return} $\{\boldsymbol{\rho}^*,\boldsymbol{v}_u^*,\boldsymbol{s}^*,\boldsymbol{q}_u^*\}=\{\boldsymbol{\rho}^{(i)},\boldsymbol{v}_u^{(i)},\boldsymbol{s}^{(i)},{\boldsymbol{q}^{(i)}_u}\}$.
\end{algorithmic}
\end{algorithm}

\begin{proposition}
In steps of Algorithm 2, the objective values of (P1) are monotonically increasing after each iteration of Algorithm 1. However, these values are non-negative.  Thus, the proposed AO algorithm is guaranteed to converge.
\end{proposition}
\begin{proof}
See Appendix \ref{sec:Proposition2}.
\end{proof}

\section{Complexity Analysis}\label{Complexity}
In this section, we investigate the complexity of Algorithm \ref{algorithm2} which consists of $3$ main steps and can be solved by the interior-point method \cite{Boyd}. Specifically, in step $2$, the complexity order for solving (P4) is $\mathcal{O}(KN + N +3K+ 1)^{3.5}$, where $KN + N +3K+ 1$ indicates the number of variables. In step $3$, the complexity of computing the UAV trajectory/velocity is $\mathcal{O}\left( I_1(2K + 3N + 1)^{3.5}\right)$, where $2K + 3N + 1$ denotes the number of variables and $I_1$ is the number of iterations. Finally, in step $4$, the reflection coefficient with the complexity of $\mathcal{O}(2K+N+1)^{3.5}$ is obtained, where $2K+N+1$ stands for the number of variables. Accordingly, the total complexity order of Algorithm \ref{algorithm2} is $\mathcal{O}( I_2( (KN + N +3K+ 1)^{3.5}+I_1(2K + 3N + 1)^{3.5}+(2K+N+1)^{3.5}))$, where $I_2$ is the number of iterations expected for convergence \cite{Boyd}. The complexity analysis provides several insights into the network design. Although the LoS links provided by UAV-carried IRS facilitate communication, the complexity of Algorithm 2 grows as the number of reflecting elements at the IRS increases. In other words, scaling up the IRS can improve EH efficiency and reduce the energy consumption of users, but this comes at the cost of higher complexity. Thus, there is a trade-off between complexity and performance gain.
\begin{table}[t]
\renewcommand{\arraystretch}{1.05}
\centering
\caption{ Simulation Parameters}
\label{table-notations}
\begin{tabular}{|l|l| }    
\hline
\textbf{Parameters}& \textbf{Values}\\\hline      
Maximum flying acceleration of the UAV, $a_\text{max}$& $4$~$\text{m}/\text{s}^2$  \\ \hline  
Maximum flying speed of the UAV, $V_\text{max}$& $20$~m/s  \\ \hline  
Maximum altitudes of the UAV, $q^z_u[n]$ & $20$~m  \\ \hline
Received antenna noise power,~$\sigma^2$ & $-140~\text{dBm/Hz}$  \\ \hline
Number of reflecting elements,~$M$ & $25$\\ \hline 
Number of discrete time slots, $\delta_T$  & $0.5$~s\\ \hline
Non-linear EH model parameters \cite{Nonlinear}&
\begin{tabular}[c]{@{}l@{}} $c_k=6400$, \\ 	$\nu_k=0.003$\end{tabular}  \\ \hline
Maximum harvested power, $M_k$ & $0.02$~Watt\\ \hline
Maximum transmit power, $P_{k}$ \cite{Jae_Cheol} &$0.1$~Watt\\ \hline
Target data rate, $R_\text{min}$ & $10$ Mbits\\ \hline
Rician factor,~$K_r$ & $10$~dB \\ \hline
Carrier frequency &  $3.4$~GHz  \\ \hline
Network bandwidth  &  $1$~MHz  \\ \hline
\end{tabular}
\end{table}

\section{Numerical Results}\label{sec5}
This section provides numerical results to evaluate the performance of our proposed scheme.~A 3D coordinate network is considered. The total number of reflecting elements at the IRS is assumed to be $M=M_xM_y$. To get a better insight into the network model, we consider two different trajectories for the IRS-UAV namely trajectory $1$ and trajectory $2$. The initial and final positions of trajectory $1$ is $(0,\:30,\:10)$~m and $(300,\:65,\:10)$~m, receptively, while for the second trajectory, it is $(0,\:45,\:10)$~m and $(300,\:45,\:10)$~m, respectively. For simplicity, we consider three users ($K = 3$) \cite{Kwan}. The location of the BS is ($150$, $50$, $8$)~m and the location of users are set as follows: $\text{u}_1=(20,\:50,\:1)$~m, $\text{u}_2=(120,\:40,\:1)$~m, and $\text{u}_3=(240,\:55,\:1)$~m. Table I gives the simulation parameters,~unless otherwise specified. In addition, we consider the maximum altitude of the UAV equal to $20$ m. This assumption is standard throughout the literature (e.g., see \cite{Huang,Mo_Kang,Wei_Wang,DSixing_Yin,Cheol_Jeong,Pan1}) and fairly realistic. In this way, the IRS-UAV flies at the lowest allowable flight altitude to obtain a higher channel gain for maximizing the harvested energy at the IRS and reflect the incident signals simultaneously. The distance-dependent path loss model is given by $
L(d) = C_0(\frac{d}{D_0})^{-\alpha}$, where $C_0=\left(\frac{3\times 10^8}{4\mathsf{f}\pi} \right)^2$ indicates the path loss at the reference distance  $D_0=1$ m \cite{Xiao2020},~$d$ is the link distance ,~and $\alpha$ denotes the path loss exponent.~More specifically,~the path loss exponents of links are assumed to be $\alpha=2.4$ \cite{Pan2}. We also  study the following comparative  baselines: 
\begin{enumerate}
\item [i)]	  Proposed scheme (Algorithm 2): Optimizing the user scheduling, UAV trajectory/velocity, and IRS reflection coefficients.  

\item[ii)]	  Baseline  1: Proposed scheme with a fixed velocity. 

\item[iii)] Baseline  2: Proposed scheme with a fixed velocity and straight trajectory where the UAV flies in a straight line from $ {{\bf{q}}_s}$ to ${{\bf{q}}_e}$.

\item[iv)] Baseline  3: Proposed scheme without IRS (No IRS).

\item[vi)] Baseline  4:  Proposed scheme without considering EH at the IRS (No EH), i.e., $\bar{\rho}[n]=1,~\forall n$. 
\end{enumerate}
\begin{figure}[t]
\centering
\includegraphics[width=3in] {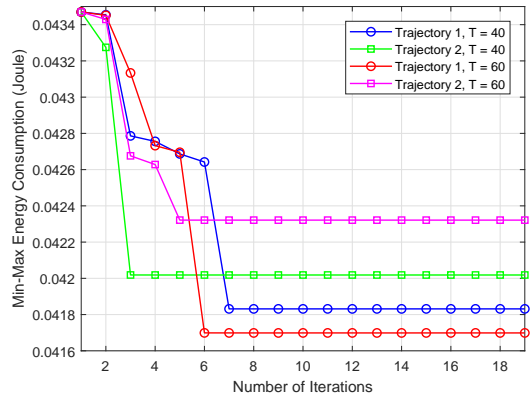}
\caption{ Convergence behavior of the proposed scheme under different network setups.}
\label{convergence}
\vspace{-3mm}
\end{figure}

\begin{figure}
\centering
\begin{minipage}[b]{.35\textwidth}
\includegraphics[width=3in]{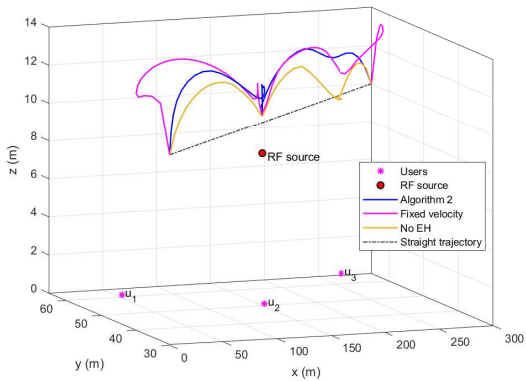}
\subcaption{Trajectory $1$ in the 3D plane.}	\label{tr1}
\end{minipage}\qquad
\begin{minipage}[b]{.35\textwidth}
\includegraphics[width=3in]{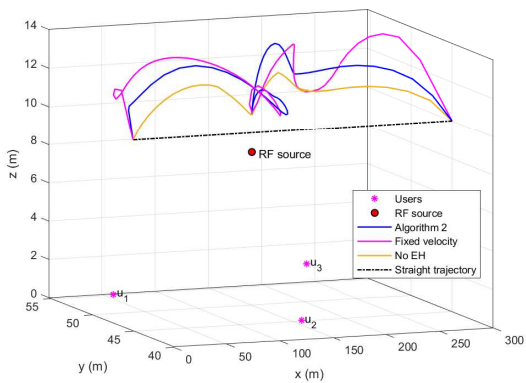}
\subcaption{Trajectory $2$ in the 3D plane.}	\label{tr2}
\end{minipage}
\begin{minipage}[b]{.35\textwidth}
\includegraphics[width=3in]{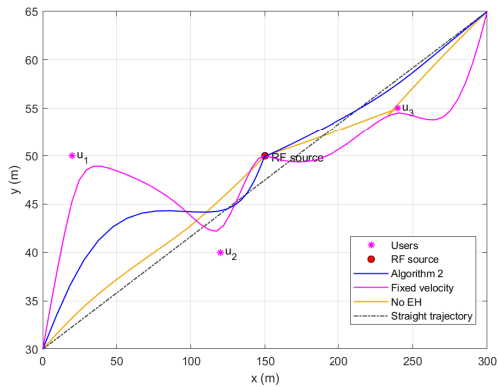}
\subcaption{Trajectory $1$ in the 2D plane.}	\label{tr11}
\end{minipage}\qquad
\begin{minipage}[b]{.35\textwidth}
\includegraphics[width=3in]{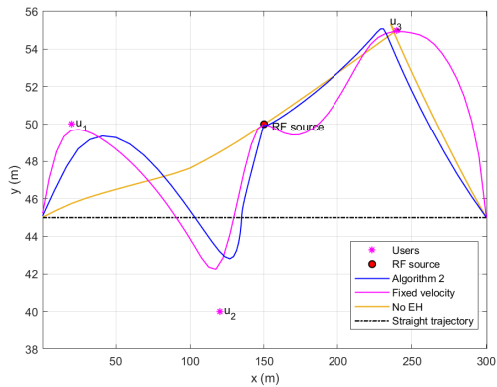}
\subcaption{Trajectory $2$ in the 2D plane.}	\label{tr22}
\end{minipage}
\caption{Trajectory behavior of the IRS-UAV under different initial and final positions for different baseline schemes.}
\label{tr}
\end{figure}

\begin{figure*}
\centering
\begin{minipage}[b]{.4\textwidth}
\centering
\includegraphics[width=3in]{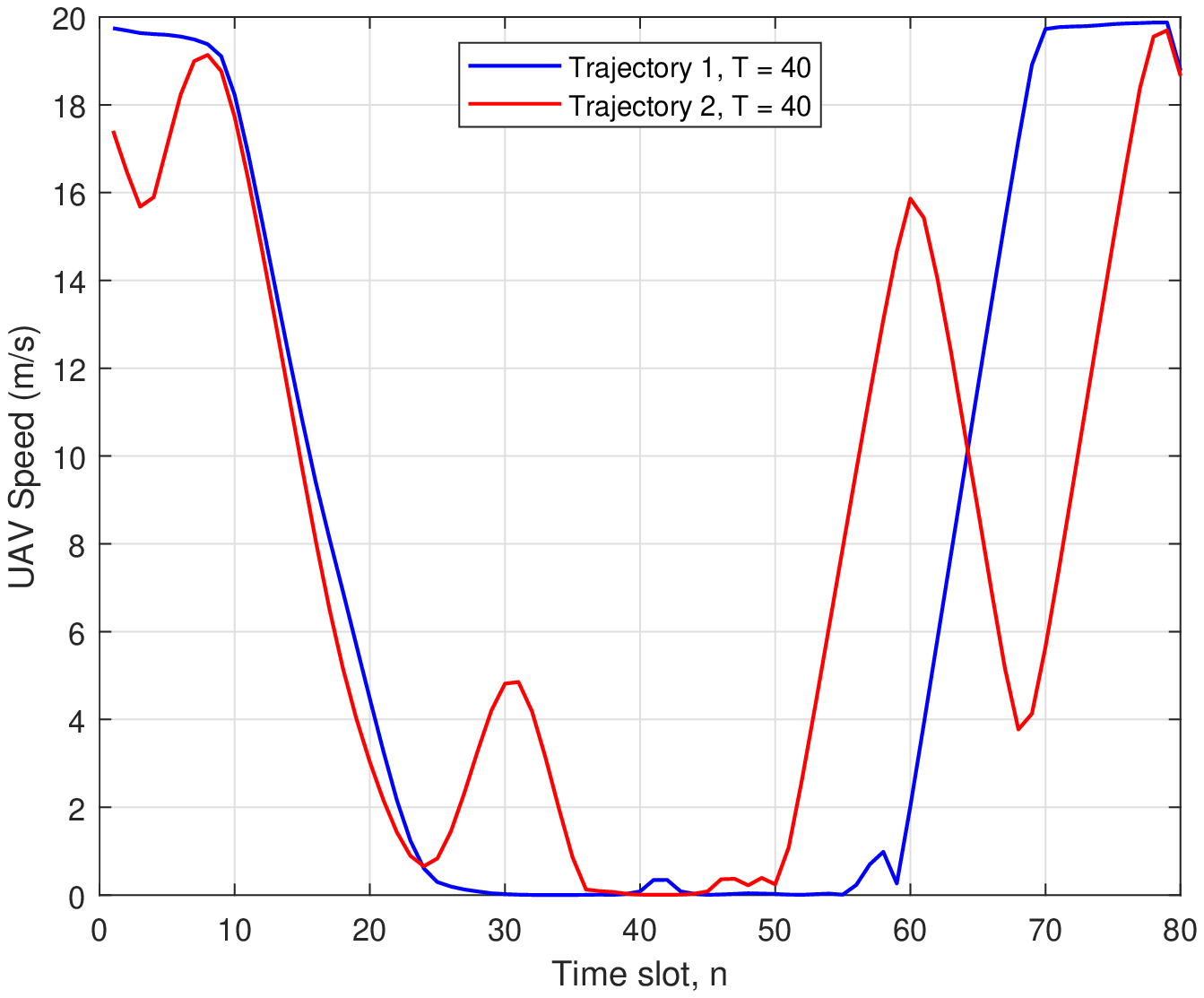}
\subcaption{UAV velocity under $T = 40$.}\label{vel1}
\end{minipage}\qquad
\begin{minipage}[b]{.4\textwidth}
\centering
\includegraphics[width=3in]{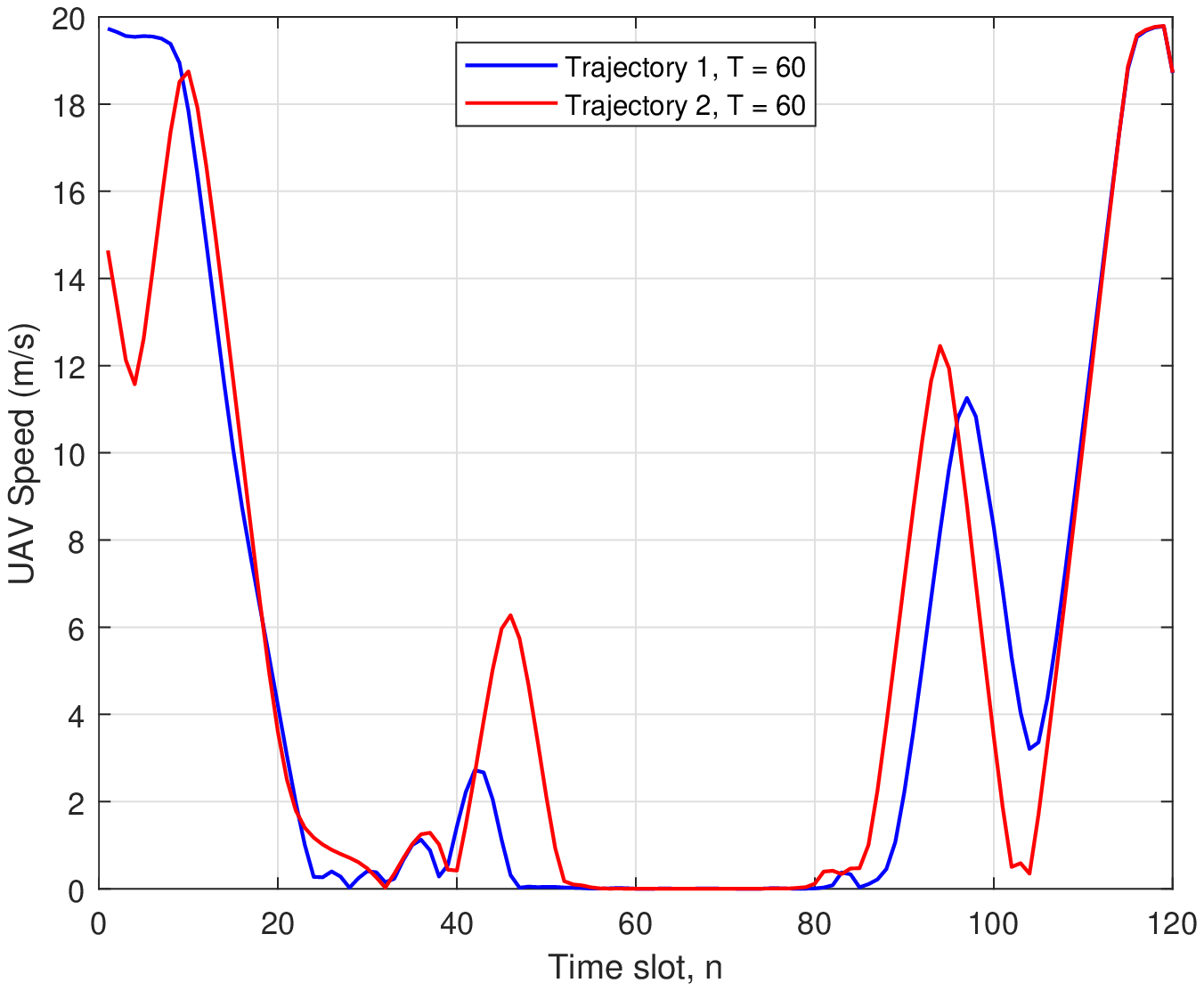}
\subcaption{UAV velocity under $T = 60$.}\label{vel2}
\end{minipage}
\caption{Velocity behavior of the IRS-UAV under different trajectories and total time slots.}
\label{vel}
\end{figure*}
\subsection{Convergence Behavior}
We study the convergence of the proposed scheme in Fig. \ref{convergence},  which shows the min-max energy consumption as a function of the number of iterations.  The plots include different network setups to highlight the convergence behavior of the proposed scheme after optimizing the user scheduling, UAV trajectory/velocity, and IRS reflection coefficient. As observed, the scheme achieves the bulk of its performance gain in just a few iterations (about seven). Thus, from the implementation point of view,  so few iterations heighten the appeal of the proposed scheme. In addition, it converges faster for trajectory $2$  than for trajectory $1$. This serves to illustrate that the  UAV trajectory significantly impacts the convergence behavior.

\subsection{Trajectory and Velocity Behavior of the IRS-UAV}
Fig. \ref{tr} and \ref{vel} demonstrate the trajectory and velocity behavior of the IRS-UAV under different trajectories for different baseline schemes and total time slots in the 3D and 2D planes\footnote{Specifically, the IRS elements are typically sub-wavelength in size (e.g., a square patch of size $\lambda/10\times \lambda/10$ to $\lambda/5\times \lambda/5$ ) to behave as scatterers without strong intrinsic directivity \cite{Henk_Wymeersch,Christos_Liaskos}. For instance, an  IRS of size $ 9 \times 9$ with a square patch of size $\lambda/5\times \lambda/5$  has a side length of approximately $ 50$ cm. A metasurface  possesses an overall areal weight density of $0.11$ g/$\text{cm}^2$ \cite{Silva}. Thus, for an IRS of size $ 9 \times 9$, the weight is roughly $0.37$ kg \cite{Torres}. For a typical UAV (model md4-1000), the maximum payload mass is $1.25$ Kg. Due to the compact size, lightweight, low energy consumption, and conformal geometry of the IRS, the UAV can readily carry it.}. In these figures, the IRS-UAV flies near the users by adjusting its trajectory with fixed or optimized velocity for all proposed schemes. For the proposed scheme where the IRS-UAV serves as a reflector, the IRS-UAV flies close to the BS to increase the performance gain in terms of the min-max energy consumption of all users. This is because flying near the BS can establish better channel conditions for the IRS-UAV, leading to better performance gains. Besides, by modifying the IRS phase shifts to align the cascaded angle of arrival (AoA) and AoD with the user-BS link,  we can reduce the maximum energy consumption of all users. The UAV adjusts its trajectory to move closer to the BS for baseline 4 where no EH is performed at the IRS. This is interesting since the number of users in the 6G is proliferating, but each may not have an LoS link with the BS. However, thanks to IRS-UAV, LoS links can be established between each user and the BS via adjusting the UAV trajectory/velocity and the IRS phase shifts.

With existing LoS links, the IRS-UAV can still help each user reduce the transmit power,  leading to improvements in the SNR at the BS. Subsequently, in Fig. \ref{vel}(a) and \ref{vel}(b), we show the velocity of the IRS-UAV under different total time slots, $T$, for trajectory $1$ and trajectory $2$. In both figures, the IRS-UAV reduces its speed by moving sufficiently close to each user and then increasing it to serve the subsequent user. In the end, it approaches the final location with a sudden increase in speed.  We observe that by increasing the total operational time, $T$, the IRS-UAV has a sufficiently large flying time to move its preferred position to serve the user near the BS. Another interesting result is that with $T=60$, the IRS-UAV flies at a lower speed to enhance the received signal strength. However, with $T=40$ the IRS-UAV requires to increase its speed to serve all users.

\begin{figure}[t]
\centering
\includegraphics[width=3in]{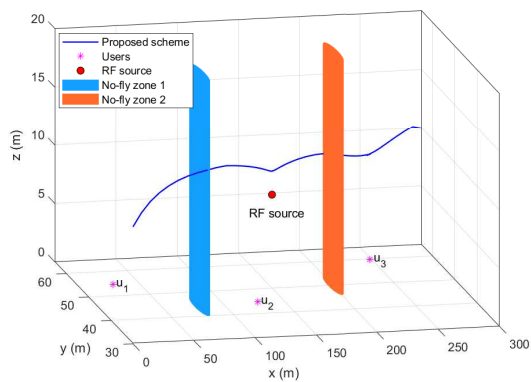}
\caption{ Behavior of trajectory $1$ in the 3D plane under the existence of No-fly zones.}
\label{No_fly_zone_new}
\vspace{-3mm}
\end{figure}	
\begin{figure*}
\centering
\begin{minipage}[h]{.4\textwidth}
\centering
\includegraphics[width=3in]{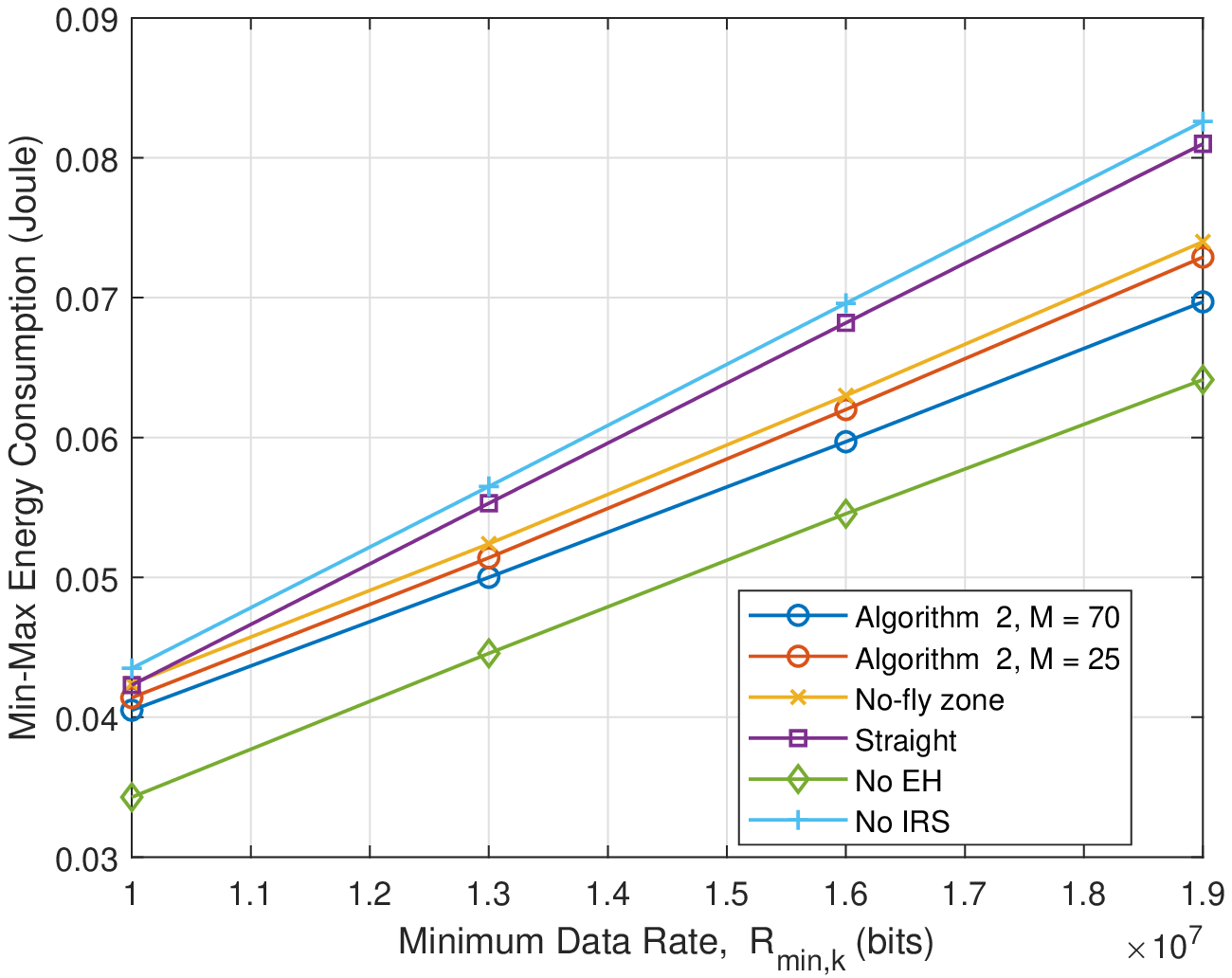}
\caption{ Min-max energy consumption vs. minimum data rate under trajectory 1 and $T = 60$.}
\label{datarate}
\end{minipage}\qquad\qquad
\begin{minipage}[h]{.4\textwidth}
\centering
\includegraphics[width=3in]{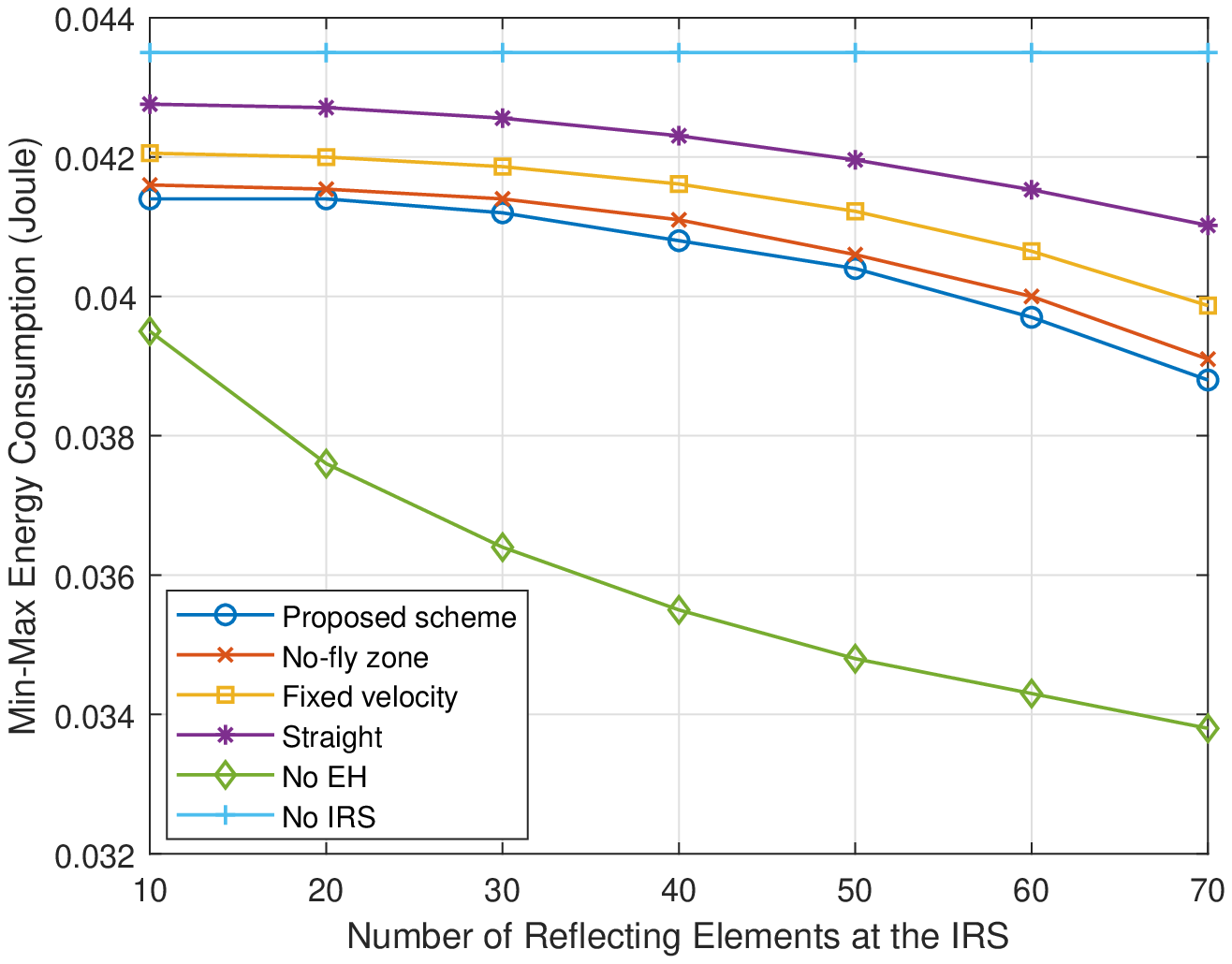}
\caption{ Min-max energy consumption vs. the number of IRS reflecting elements under trajectory 1 and $T = 60$.}
\label{ref}
\end{minipage} 
\vspace{-3mm}
\end{figure*}
On the other hand, in practice, there may be areas (No-fly zones) such that UAVs are strictly prohibited to fly over them. Therefore, these areas need to be taken into consideration for UAV trajectory planning. Accordingly, we consider two non-overlapped No-fly zones in which the trajectory of the UAV needs to satisfy the following constraints:
\begin{align}
&\quad~C_{13}:~{\|\mathbf{q}_u[n]-\mathbf{q}_{\text{NF},1}[n]\|^2}\geq (Q_{\text{NF},1})^2,~\forall n,\nonumber\\
&\quad~C_{14}:~{\|\mathbf{q}_u[n]-\mathbf{q}_{\text{NF},2}[n]\|^2}\geq (Q_{\text{NF},2})^2,~\forall n,\nonumber
\end{align}
where $\mathbf{q}_{\text{NF},1}[n]=(80,\:44,\:20)$~m and $\mathbf{q}_{\text{NF},1}[n]=(200,\:50,\:20)$~m are the location of the No-fly zones. Also, $Q_{\text{NF},1}=4$~m and $Q_{\text{NF},1}=4$~m are the radius of the No-fly zones. As these constraints are non-convex, we resort to the first-order Taylor expansion to obtain the lower bounds similar to (\ref{key}). In Fig \ref{No_fly_zone_new}, one can observe the IRS-UAV adjusts its trajectory in such a way as to pass these two No-fly zones.

\subsection{Min-Max Energy Consumption vs. Minimum Data Rate}

Fig. \ref{datarate} investigates the min-max energy consumption of all users versus the minimum data rate requirement of each user for different schemes.  We see that by increasing the minimum data rate, the min-max energy consumption grows significantly. The reason is that for large minimum data rates, each user requires higher transmit power,  leading to high energy consumption and accordingly deteriorating the network performance.

Besides, IRS without considering energy harvesting has a better performance compared to the proposed scheme since no additional power is needed to empower the IRS reflecting elements. Finally, the proposed scheme with a straight trajectory has worse performance compared to other cases except the No IRS case. This is because without optimizing trajectory, the IRS-UAV cannot fly close or even stay above the users with better channel conditions which accordingly increase the transmit power of the user to satisfy the minimum data rate requirement. The No IRS case has the worst performance among other schemes. On the other hand, we observe almost the same result by considering No-fly zones, which indicates that our proposed scheme can preserve the performance gain even in the non-flying zones.

\subsection{Min-Max Energy Consumption vs. the number of IRS  Elements}
Fig. \ref{ref} shows the min-max energy consumption of all users versus the number of IRS reflecting elements. As they increase, the energy consumption of all users decreases monotonically. The elements also create a powerful reflective channel link, extending the communication range by adjusting the trajectory and phase shifts of the UAV and IRS, respectively. However, the proposed scheme has a limited impact on the performance gain for a few reflecting elements compared to the No IRS case, showing that the number of reflecting elements is a bottleneck. This figure also reveals an exciting result: although many reflecting elements will increase the EH requirements, the IRS-UAV can lower the users' energy consumption, highlighting the critical importance of joint optimization of the trajectory and the phase shifts. We can readily increase the number of IRS  reflecting elements at a much lower cost since these elements eschew RF chains. The proposed scheme is also superior to that of fixed velocity and straight trajectory. In addition, no  EH at the IRS option achieves better performance than our proposed scheme since the IRS in this case reflects incident signals completely without performing EH. The IRS  reflecting elements are not entirely passive, and their energy consumption cannot be neglected. Hence, we must consider the energy consumption of the IRS, especially when it is mounted on a UAV. This is because the UAV should supply the energy consumption of the IRS. To reduce energy consumption, EH is a helpful solution that is investigated in this paper. 

\section{Conclusion}\label{sec6}
This paper analyzed the user scheduling and trajectory optimization of the UAV-carried IRS network. The overall system model constitutes the following features. The IRS is mounted on the UAV. The UAV moves along a 3D trajectory.  The IRS can harvest its operational power through the uplink signals from users. It helps the uplink data transmission from users to the BS by establishing reflected channel links. The EH process is modeled as a  realistic non-linear  EH circuit.  With the aid of this system model, we minimized the maximum energy consumption of users by joint optimization of user scheduling, UAV trajectory/velocity, and IRS phase shifts/reflection coefficient. To solve this optimization problem, we first derived closed-form  IRS phase shifts and then obtained other optimization variables by resorting to the AO algorithm. The efficiency of the proposed scheme was investigated relative to the following benchmarks: fixed velocity UAV, straight trajectory, No EH at the IRS, and No IRS. Simulations demonstrate the gains of the proposed scheme over the benchmarks. For instance, with a  $50$-element IRS, min-max energy consumption can be as low as $0.0404$ J,  a $7.13$\%  improvement over the No IRS case  (achieving $0.0435$ J). We also show that IRS-UAV without EH performs best at the cost of circuit power consumption of the IRS (a $20$\% improvement over the No IRS case). To our knowledge, this is the first manuscript to design user scheduling and trajectory optimization for IRS-UAV networks. The results of this study suggest a number of new avenues for research.  First, further research is necessary for the case of users with multiple antennas. That opens up the possibility of beamforming and other MIMO techniques to enhance the network performance.  Second, one can minimize the maximum energy consumption of all users under imperfect CSI scenarios and develop robust algorithms against CSI imperfections. There are many more future directions, and this list is not exhaustive.

\appendix

\subsection{Proof of Proposition 1}\label{sec:Proposition1}		
The approximation in (\ref{399}) produces a tight lower bound of $\Gamma_k[n]$. This is because $\bar{\Gamma}_k[n]$ is a concave function. The gradient
of $\bar{\Gamma}_k[n]$ is a supper-gradient given by $	\bar{\Gamma}_k[n] \leq 	\Psi_k[n]$ \cite{zargari3}, where $\mathcal{H}=\{  u[n],r_k[n],t_k[n],\boldsymbol{q}_u,\boldsymbol{v}_u,\chi\}$ indicates the set of feasible solutions at iteration $t$. Besides, the equality holds when $\mathcal{H}=\mathcal{H}^{(t-1)}$, which confirms the tightness of the lower bound. By denoting $\Psi_k[n]$ at iteration $t$ as $\Psi^{(t)}_k[n]$, we have the following relations: $\Psi^{(t+1)}_k[n] \geq\Psi^{(t)}_k[n]$ which leads to the increase of the objective values of (P7) after each iteration of Algorithm 1. Accordingly, by solving the convex lower bound in (P4), the iterative-based SCA algorithm creates a sequence of feasible solutions, i.e., $\mathcal{H}^{(t+1)}$, which is monotonically increasing over each iteration. Thus, it is guaranteed to converge.

\subsection{Proof of Theorem 1}\label{sec:Theorem1}

According to the Jensen's inequality \cite{Boyd}, i.e., $\mathbb{E} \{\log_2 (1 + f(z)\} \leq \log_2 (1 + \mathbb{E} \{f(z)\})$, the following inequality holds
\begin{align}\label{Jensen}
&\mathbb{E}\{R_{k}[n]\}\leq s_k[n] \log_{2}\left( 1+\frac{ P_k g_k}{\sigma^{2}}\right),
\end{align}
where $g_k = \mathrm{E}\left\lbrace  \left|  {h}_{k,b}[n]+\bar{\rho}[n]\mathbf{h}^H_{k,u}[n]{\mathbf{\Theta}}[n]\mathbf{h}_{bu}[n]\right| ^2\right\rbrace$. Since the NLoS components of each link, i.e., ${h}_{k,u}^{\mathrm{NLoS}}[n]$, $\mathbf{h}_{k,b}^{\mathrm{NLoS}}[n]$, and $\mathbf{h}_{bu}^{\mathrm{NLoS}}[n]$ are independent of each other, the square term in $g_k$ can be written as $|b|^2+\mathrm{E}\{|a|^2\}+\mathrm{E}\{|c|^2\}+\mathrm{E}\{|d|^2\}+ \mathrm{E}\{|e|^2\}$, 	where $a=\sqrt{\frac{\beta_{k,b}}{K_{k,b}+1}}h_{k,b}^{\mathrm{NLoS}}[n]$ and
\begin{align}
&b=\sqrt{\frac{K_{k,b}\beta_{k,b}}{K_{k,b}+1}}h_{k,b}^{\mathrm{LoS}}[n]+\sqrt{\frac{K_{k,u}K_{bu}\beta_{bu}[n]\beta_{k,u}[n]}{(K_{k,u}+1)(K_{bu}+1)}}\nonumber\\& \times\bar{\rho}[n](\mathbf{h}_{k,u}^{\mathrm{LoS}}[n])^{H}{\mathbf{\Theta}}[n]\mathbf{h}_{bu}^{\mathrm{LoS}}[n],\nonumber\\
&c=\sqrt{\frac{K_{bu}\beta_{bu}[n]\beta_{k,u}[n]}{(K_{k,u}+1)(K_{bu}+1)}}\bar{\rho}[n](\mathbf{h}_{k,u}^{\mathrm{NLoS}}[n])^{H}{\mathbf{\Theta}}[n]\mathbf{h}_{bu}^{\mathrm{LoS}}[n],\nonumber\\
&d=\sqrt{\frac{K_{k,u}\beta_{bu}[n]\beta_{k,u}[n]}{(K_{k,u}+1)(K_{bu}+1)}}\bar{\rho}[n](\mathbf{h}_{k,u}^{\mathrm{LoS}}[n])^{H}{\mathbf{\Theta}}[n]\mathbf{h}_{bu}^{\mathrm{NLoS}}[n],\nonumber\\
&e=\sqrt{\frac{\beta_{bu}[n]\beta_{k,u}[n]}{(K_{k,u}+1)(K_{bu}+1)}}\bar{\rho}[n](\mathbf{h}_{k,u}^{\mathrm{NLoS}}[n])^{H}{\mathbf{\Theta}}[n]\mathbf{h}_{bu}^{\mathrm{NLoS}}[n].
\end{align}
By taking the term $c$ into consideration, this can be calculated as follows:
\begin{align}
c&=\frac{\bar{\rho}^2[n]K_{bu} \beta_{bu}[n]\beta_{k,u}[n]}{(K_{k,u}+1)(K_{bu}+1)}(\mathbf{h}_{bu}^{\mathrm{LoS}}[n])^{H}({\mathbf{\Theta}}[n])^{H}\nonumber\\&\times\mathrm{E}\left\lbrace \mathbf{h}_{k,u}^{\mathrm{NLoS}}[n](\mathbf{h}_{k,u}^{\mathrm{NLoS}}[n])^{H}\right\rbrace {\mathbf{\Theta}}[n]\mathbf{h}_{bu}^{\mathrm{LoS}}[n].
\end{align}
By adopting the following equalities: 
\begin{align}
&\mathrm{E}\lbrace \mathbf{h}_{k,u}^{\mathrm{NLoS}}[n](\mathbf{h}_{k,u}^{\mathrm{NLoS}}[n])^{H}\rbrace=\mathbf{I}_{M_xM_y},\nonumber\\&(\mathbf{h}_{bu}^{\mathrm{LoS}}[n])^{H}\mathbf{h}_{bu}^{\mathrm{LoS}}[n]=M_xM_y, ~({\mathbf{\Theta}}[n])^{H}{\mathbf{\Theta}}[n]=\mathbf{I}_{M_xM_y},
\end{align}
we can obtain $c=\frac{\bar{\rho}^2[n]M_xM_y K_{bu}\beta_{bu}[n]\beta_{k,u}[n]}{(K_{k,u}+1)(K_{bu}+1)}$. Also, for other terms, we have the following equations:
\begin{align}
&a=\frac{ \beta_{k,b}}{K_{k,b}+1},~d=\frac{\bar{\rho}^2[n]K_{k,u}M_xM_y\beta_{bu}[n]\beta_{k,u}[n]}{(K_{k,u}+1)(K_{bu}+1)},\nonumber\\&e=\frac{\bar{\rho}^2[n]M_xM_y\beta_{bu}[n]\beta_{k,u}[n]}{(K_{k,u}+1)(K_{bu}+1)}.
\end{align} 
Thus, the proof is completed.

\subsection{Proof of Proposition 2}\label{sec:Proposition2}	

As the phase shifts are obtained in closed-form expression, the main problem is divided into three subproblems which optimize user scheduling ($\boldsymbol{s}$), UAV trajectory/velocity ($\boldsymbol{q}_u,\boldsymbol{v}_u$), and IRS reflection coefficient ($\boldsymbol{\rho}$), via solving problems (P4), (P7), and (P9), while keeping the other two blocks of variables fixed. Let us define $f(\boldsymbol{s},\boldsymbol{q}_u,\boldsymbol{v}_u,\boldsymbol{\rho})$ as a function of $\boldsymbol{s}$, $\boldsymbol{q}_u$, $\boldsymbol{v}_u$, and $\boldsymbol{\rho}$ for the objective value of (P1). First, in step $2$ of Algorithm 2 with fixed variables $\boldsymbol{q}_u$, $\boldsymbol{v}_u$, and $\boldsymbol{\rho}$, problem (P1) is a LP problem and $\boldsymbol{s}^{(i\!+\!1)}$ is the optimal solution that maximize the value of the objective function. Accordingly, we have
\begin{equation}\label{43}
f(\boldsymbol{s}^{(i\!+\!1)},\boldsymbol{q}_u^{(i)},\boldsymbol{v}^{(i)}_u,\boldsymbol{\rho}^{(i)})\geq f(\boldsymbol{s}^{(i)},\boldsymbol{q}^{(i)}_u,\boldsymbol{v}^{(i)}_u,\boldsymbol{\rho}^{(i)}).
\end{equation}
Next, in step $3$ of Algorithm 2, $\boldsymbol{q}_u^{(i\!+\!1)}$ and $\boldsymbol{v}^{(i\!+\!1)}_u$ are the suboptimal UAV trajectory/velocity with given variables $\boldsymbol{s}^{(i\!+\!1)}$ and $\boldsymbol{\rho}^{(i)}$ to maximize $f$ via solving (P9). Thus, it guarantees that
\begin{equation}
f(\boldsymbol{s}^{(i\!+\!1)},\boldsymbol{q}_u^{(i\!+\!1)},\boldsymbol{v}^{(i\!+\!1)}_u,\boldsymbol{\rho}^{(i)})\geq f(\boldsymbol{s}^{(i\!+\!1)},\boldsymbol{q}^{(i)}_u,\boldsymbol{v}^{(i)}_u,\boldsymbol{\rho}^{(i)}).
\end{equation}
Finally, in step $4$ of Algorithm 2 with the given $\boldsymbol{s}^{(i\!+\!1)}$, $\boldsymbol{q}_u^{(i\!+\!1)}$, and $\boldsymbol{v}^{(i\!+\!1)}_u$, problem (P9) is solved to obtain a sub-optimal solution for $\boldsymbol{\rho}^{(i)}$, which yields:
\begin{equation}\label{45}
f(\boldsymbol{s}^{(i\!+\!1)},\boldsymbol{q}_u^{(i\!+\!1)},\boldsymbol{v}^{(i\!+\!1)}_u,\boldsymbol{\rho}^{(i\!+\!1)})\geq f(\boldsymbol{s}^{(i\!+\!1)},\boldsymbol{q}^{(i\!+\!1)}_u,\boldsymbol{v}^{(i\!+\!1)}_u,\boldsymbol{\rho}^{(i)}).
\end{equation}
According to \eqref{43}--\eqref{45}, we can conclude that
\begin{equation} 
f(\boldsymbol{s}^{(i\!+\!1)},\boldsymbol{q}_u^{(i\!+\!1)},\boldsymbol{v}^{(i\!+\!1)}_u,\boldsymbol{\rho}^{(i\!+\!1)})\geq f(\boldsymbol{s}^{(i\!+\!1)},\boldsymbol{q}^{(i\!+\!1)}_u,\boldsymbol{v}^{(i\!+\!1)}_u,\boldsymbol{\rho}^{(i\!+\!1)}),
\end{equation}
which indicates that the objective values of (P1) are monotonically increasing after each iteration of Algorithm 2.  Meanwhile,  the objective values of (P1) are non-negative. As a result, the proposed AO algorithm is guaranteed to converge. On the other hand, based on the fact that the initial point of each iteration is the ending point of the previous one, the algorithm continues running to achieve a better solution in each iteration. In other words, the objective function increases in each iteration or remains unchanged until the convergence is satisfied. Thus, the proof is completed.

\bibliographystyle{ieeetr}
\bibliography{ref}	

\end{document}